\documentclass[11pt]{article}

\usepackage[round]{natbib}

\RequirePackage{amsthm,amsmath,amsfonts,amssymb}
\usepackage{bm} 
\usepackage{mathrsfs} 
\usepackage{color}
\usepackage{caption}
\usepackage{mathtools} 
\usepackage{enumerate}
\usepackage{url}
\usepackage[linesnumbered,ruled,vlined]{algorithm2e}
\usepackage{setspace}

\def\cX{\mathcal{X}}
\def\cN{\mathcal{N}}
\def\Edge{\mathcal{E}}
\def\bbN{\mathbb{N}}
\def\bbM{\mathbb{M}}
\def\bbE{\mathbb{E}}
\def\bbR{\mathbb{R}}
\def\bbP{\mathbb{P}}
\def\ind{\bm{1}}

\def\gap{\mathrm{Gap}}
 
\def\Var{\mathrm{Var}}

\def\hbbR{\bbR_{+}}
\def\Var{\mathrm{Var}}
\def\TGS{ \mathrm{TGS}}
\def\WTGS{ \mathrm{WTGS}}

\def\tpi{\tilde{\pi}}
\def\cT{\mathsf{T}}
\def\Anc{\mathcal{A}}
\def\hsq{h_{0.5}}
\def\htgs{h_{+1}}
\def\hmin{h_{\wedge 1}} 
\def\bpi{\bar{\pi}}
\def\cG{\mathsf{G}}
\def\cY{\mathcal{Y}}

\def\sgamma{\gamma_0}
\def\bP{\bar{P}}
\def\se{\bm{e}}
\def\tmix{\tau_{\mathrm{mix}}}
\def\pmin{\pi_{\mathrm{min}}}

\def\bcN{\overline{\mathcal{N}}}  
\newcommand{\tg}{\tilde{g}}
\newcommand{\cpi}{\check{\pi}}
\newcommand{\comega}{\check{\omega}}
\newcommand{\xmap}{\hat{x}}
\newcommand{\Tmap}{T_{\hat{x}}}

\DeclareMathOperator*{\argmax}{arg\,max}
\DeclarePairedDelimiter{\norm}{\lVert}{\rVert}

\newtheorem{lemma}{Lemma}

\newtheorem{theorem}{Theorem}
\theoremstyle{definition}
\newtheorem{definition}{Definition}
\newtheorem{remark}{Remark}
\newtheorem{example}{Example}

\newtheorem{assumption}{Assumption}

\newtheorem{alg}{Algorithm}
  
\def\TITLE{Rapid Convergence of Informed Importance Tempering} 
\title{\TITLE}
\author{Quan Zhou$^1$ \and Aaron Smith$^2$}
\date{$^1$Department of Statistics, Texas A\&M University \\ 
$^2$Department of Mathematics and Statistics, University of Ottawa}
  
\usepackage[a4paper, left=1.25in,right=1.25in,top=1.5in,bottom=1.5in,%
            footskip=.25in]{geometry}
\renewcommand{\MakeUppercase}[1]{#1}
\usepackage{caption}
\begin{document}

\maketitle

\begin{abstract} 
Informed Markov chain Monte Carlo (MCMC) methods have been proposed as scalable solutions to Bayesian posterior computation on high-dimensional discrete state spaces, but   theoretical results about their convergence behavior in general settings are lacking. 
In this article, we propose a class of MCMC schemes called informed importance tempering (IIT), which combine importance sampling and informed local proposals, and derive generally applicable spectral gap bounds for IIT estimators. 
Our theory shows that IIT samplers have remarkable scalability when the target posterior distribution concentrates on a small set.  
Further, both our theory and numerical experiments demonstrate that the informed proposal should be chosen with caution: the performance of some proposals may be very sensitive to the shape of the target distribution. We find that the ``square-root proposal weighting'' scheme tends to  perform well in most settings.  
\end{abstract}
  
\section{\MakeUppercase{Introduction}}\label{sec:intro}
Bayesian inference provides a flexible framework for modeling complex data and assessing  uncertainty of model selection and parameter estimation, but these advantages often come at the cost of intensive posterior computation. 
In recent years, various informed Markov chain Monte Carlo (MCMC)  methods have been proposed for sampling from discrete state spaces, which are particularly useful for model selection problems and have been shown numerically to scale well to high-dimensional data~\citep{titsias2017hamming, zanella2019scalable,  zanella2020informed, griffin2021search}; see~\citet{zhou2021local} and the references therein.\footnote{
We only consider discrete-state-space problems in this article. On continuous state spaces, gradient-based informed MCMC samplers, such as Metropolis-adjusted Langevin algorithm and Hamiltonian Monte Carlo~\citep{roberts1998optimal,girolami2011riemann}, have become almost standard solutions to posterior computation, and their theoretical properties are well understood.}
These methods require evaluating the posterior distribution locally in each iteration  so that neighboring states with larger posterior probabilities are more likely to be proposed. 
A theoretical guarantee for the scalability of informed MCMC samplers was recently obtained by~\citet{zhou2021local}, who proved that their informed  Metropolis-Hastings (MH) algorithm for high-dimensional variable selection can achieve a mixing rate that is independent of the number of variables.  
 
In this article,  we consider another approach to making use of informed proposals that is not based on MH algorithm:  accept all the proposed moves and use importance weights to correct for the proposal bias. 
We call this scheme informed importance tempering (IIT). 
Since an informed proposal distribution usually has a shape similar to the local posterior landscape, a combination of informed proposals and importance sampling can sometimes be strikingly efficient. 
One example in the literature is the tempered Gibbs sampler (TGS) for variable selection devised by~\citet{zanella2019scalable}, which has largely motivated the general framework to be proposed in this work.  
The convergence rate of  IIT estimators can be measured by the spectral gap of a continuous-time Markov chain, which enables us to use Markov chain theory to investigate the complexity of IIT schemes in general high-dimensional settings. 
We first consider the case where the target posterior distribution satisfies a unimodal condition and concentrates on one state, which, for model selection problems, can be interpreted as a strong notion of statistical consistency. 
Our theory suggests that in this setting,  IIT schemes with locally balanced proposals (see Definition~\ref{def:balance}), including TGS,  have superior scalability.  
Next, we relax the unimodal condition by assuming the posterior mass concentrates on a small set. It turns out that then TGS may lose its advantage completely, while some other IIT samplers still perform well.  
More interestingly, both our theory and numerical study support the use of  the square root of the posterior probability as the proposal weight, which renders the algorithm much more robust than other weighting schemes such as the one used by TGS. 
Finally, we extend our results to general decomposable target distributions. The spectral gap bound we obtain suggests that another way to achieve robustness is to use bounded proposal weights (see Remark~\ref{rmk:bounded}), which echoes the LIT-MH algorithm (MH with locally informed and thresholded proposals) developed by~\citet{zhou2021local}. 
 
The rest of the paper is structured as follows. We introduce the notation and the IIT algorithm in Section~\ref{sec:mcmc}. 
Spectral gap bounds for IIT samplers are presented in Section~\ref{sec:main}.  
Section~\ref{sec:sim} presents two simulation studies, and Section~\ref{sec:disc} concludes the paper with some discussion on the literature and how to use IIT in practice. All proofs are deferred to the supplement. 

\section{\MakeUppercase{Informed Importance Tempering}}\label{sec:mcmc}

\subsection{Notation, Problem Setup and Preliminaries} \label{sec:notation}
Throughout this work,  we use $\cX = \cX(p)$ to denote a finite set, where $p \in (1, \infty)$ is a parameter describing the problem size, and use $\pi$ to denote a  target posterior distribution with support $\cX$; that is, $\pi(x) > 0$ for every $x$.  
Our goal is to approximate $\pi$ by sampling when $\pi$ is known only up to a normalizing constant. 
For convenience of interpretation, we treat $\cX$ as a set of candidate models in a model selection problem with $p$ variables.  
The cardinality of  a set is denoted by $| \cdot |$. We are mostly interested in the cases where $|\cX|$ grows super-polynomially with $p$.  
For any function $f \colon \cX \rightarrow \bbR$, let $\bbE_\pi[ f ] = \sum_{x \in \cX} f(x) \pi(x)$ denote its expectation. 
A stochastic matrix (i.e., transition matrix) is denoted by $P$ or $K$, which is treated as a mapping from $\cX^2$ to $[0, 1]$. 
Similarly, a transition rate matrix is denoted by a mapping  $Q  \colon \cX^2 \rightarrow \bbR$. 
  
Suppose  that  a neighborhood mapping  $\cN \colon \cX \rightarrow 2^{\cX}$ is given such that $x \notin \cN(x)$ for each $x$; the set $\cN(x)$ is referred to as the neighborhood of $x$. 
Let $\bbM(\cX, \cN)$   denote the collection of all stochastic matrices $K$ with state space $\cX$ such that $K(x, y) > 0$  if and only if  $y \in \cN(x)$.    
We make two assumptions on $\cN$.  First,  $\cN$ is symmetric, which means that $x \in \cN(y)$ implies $y \in \cN(x)$. 
Second,  any $K \in \bbM(\cX, \cN)$ is irreducible. 
The term ``neighborhood''  also connotes that $| \cN(\cdot)|$ tends to be much smaller than $| \cX |$, though we will not formally impose this assumption until Section~\ref{sec:main}.    
Any $K \in \bbM(\cX, \cN)$ can be used as the proposal scheme for constructing an MH algorithm targeting  $\pi$. If $K(x, y) =  |\cN(x)|^{-1}$ for each $y \in \cN(x)$,  we refer to the corresponding algorithm as RWMH (RW: random walk).   
Details of MH algorithms are omitted since they are not the focus of this work. 

\begin{example} \label{ex:var1}  
Variable selection will be used as a running example. 
Define $\cX = \{  x \in  \{0, 1\}^p \colon  \norm{x}_1 \leq s \}$ for some $s = s(p) > 0$, where $\norm{x}_1$ denotes the $L^1$-norm. Each $x \in \cX$ represents a sparse linear regression model such that the $i$-th variable is included if and only if $x_i = 1$.  
In high-dimensional settings, we usually let $s \rightarrow \infty$, so $|\cX|$ is super-polynomial in $p$.  
For each $x \in \cX$, we can define ``add'', ``delete'' and ``swap'' neighborhoods by 
\begin{align*}
\cN_{\rm{add}}(x) &=\{y \in \cX \colon  \norm{x - y}_1  = 1, \norm{y}_1 = \norm{x}_1 + 1 \}, \\
\cN_{\rm{del}}(x) &=\{y \in \cX \colon  \norm{x - y}_1  = 1, \norm{y}_1 = \norm{x}_1 - 1 \},  \\
\cN_{\rm{swap}}(x) &= \{ y \in \cX\colon  \norm{x - y}_1  = 2, \norm{y}_1 = \norm{x}_1 \}. 
\end{align*}
Let  $\cN^1(x) = \cN_{\rm{add}}(x) \cup \cN_{\rm{del}}(x)$ and $\cN^2(x) = \cN^1(x) \cup \cN_{\rm{swap}}(x) $.  Both $\cN^1$  and $\cN^2$ satisfy our assumptions. 
For most MH algorithms on $\cX$ used in practice, the proposal matrix belongs to either $\bbM(\cX, \cN^1)$ or  $\bbM(\cX, \cN^2)$.   
\end{example} 

We say a proposal scheme $K \in \bbM(\cX, \cN)$ is informed  if the proposal probability depends on  the un-normalized posterior.   
We follow~\citet{zanella2020informed} to consider informed proposals that can be written as 
\begin{equation}\label{eq:def.h}
K_h(x, y)  = \frac{\ind_{\cN(x)}(y)}{Z_h(x)}    h \left( \frac{\pi(y)}{\pi(x)} \right) ,  
\end{equation}
for some $h \colon  \hbbR \rightarrow \hbbR$,   
where  $\hbbR = (0, \infty)$, $\ind$ is the indicator function,   and the normalizing constant $Z_h$ is  
\begin{equation}\label{eq:def.Zh}
    Z_h(x) = \sum_{y \in \cN(x)} h \left( \frac{\pi(y)}{\pi(x)} \right). 
\end{equation}
The function $h$ determines how the proposal weight of each neighboring state is calculated.  
In most cases,  we want $h$ to be non-constant  and non-decreasing.   
One simple choice    is  $h(u) = u^a$  for some $a > 0$, which favors neighboring states with larger posterior probabilities.  
\citet{zanella2020informed} proposed to use a  ``balancing function''. 
\begin{definition}\label{def:balance}
If  $h(u) = u \, h(1 / u)$ for any $u \in \hbbR$, we say $h$ is a  balancing function  and  $K_h$ is  a locally balanced proposal. 
\end{definition}

\begin{remark}\label{rmk:balance}
The class of balancing functions is  very rich. 
Let $h, h'$ be  two balancing functions and $g\colon \hbbR \rightarrow \hbbR$ be an arbitrary non-negative function. 
We can define new balancing functions $h_1, h_2, h_3$ by letting $h_1 = a h + a' h'$ for any $a, a' \geq 0$, 
$h_2(u) =  h(u) g(u) g(1 / u)$ and $h_3(u) = \min\{ g(u),  u  g(1/u) \}$ (or $h_3(u) = \max \{ g(u),  u  g(1/u) \}$). 
The last property was used in~\citet{zanella2020informed} to compare different informed proposals. 
\end{remark}

\begin{example}\label{rmk:three.h}
Three balancing functions will be considered in our analysis, which we denote by 
\begin{equation}\label{def:three}
\hsq(u) = \sqrt{u},  \;  \hmin(u) = 1 \wedge u,  \; \htgs(u) = 1 + u. 
\end{equation}
Note that $\htgs$ behaves just like $h(u) = 1 \vee u$, since $(1 + u)/2 \leq 1 \vee u < 1 + u$ for any $u > 0$.  So $\htgs,  \hmin$ and $\hsq$ represent three very different proposal weighting strategies:  $\hmin$ treats any $y \in \cN(x)$ with $\pi(y) \geq \pi(x)$ equally, while $\htgs$ assigns roughly the same weight to any $y \in \cN(x)$ with $\pi(y) < \pi(x)$. 
Only $\hsq$ always ``makes full use'' of the knowledge about the local posterior landscape. 
\end{example}

We will always assume that the time needed to evaluate $\pi(x)$ (up to a normalizing constant)  for any $x \in \cX$  is $O(1)$, which is also the complexity of one iteration of RWMH.   
Generating a sample from an informed proposal distribution $K_h(x, \cdot)$ typically requires evaluating $\pi$ in the entire neighborhood $\cN(x)$, so the time complexity of one informed iteration at state $x$ is  $O(| \cN(x)| )$.  
This needs to be taken into account when we compare any informed MCMC method with RWMH. 

\subsection{Algorithm}\label{sec:iit} 
Let $x^{(1)}, x^{(2)}, \dots, x^{(t)}$ denote $t$ samples generated from an irreducible Markov chain with stationary distribution  $\tpi$.   
For any function $f \colon \cX \rightarrow \bbR$, we can estimate $\bbE_\pi[f]$   using the  self-normalized importance sampling estimator
\begin{equation}\label{eq:sni}
\hat{f}(t, \omega)  =   \frac{ \sum_{k=1}^t f(x^{(k)}) \omega (x^{(k)}) }{  \sum_{k=1}^t   \omega (x^{(k)}) }, 
\end{equation}   
where $\omega(x) = \pi(x) / \tpi(x)$ is called the importance weight of the sample $x$. 
Note that we do not need to know the normalizing constants of  $\pi$ and $\pi_h$, which are canceled out.  
Such an MCMC importance sampling scheme is commonly used with  $\tpi(x) \propto \pi(x)^{1 / T}$, where the ``temperature'' $T$ can be treated as either fixed or an auxiliary random variable~\citep{jennison1993discussion, neal1996sampling}.  \citet{gramacy2010importance} called this method importance tempering (IT), and they noted that successful applications of IT schemes were surprisingly rare. 
Recently,  \citet{zanella2019scalable} proposed the TGS algorithm by combining IT with Gibbs sampling, and a weighted version of TGS demonstrated excellent performance  in high-dimensional variable selection. The great success of TGS can be mainly attributed to its informed choice of the coordinate to update. 
The IIT algorithm we propose generalizes the main idea of TGS. 
 
\begin{alg}[Informed importance tempering]\label{alg:iit}
Let $x^{(0)} \in \cX$ and $h \colon  \hbbR \rightarrow \hbbR$ be given. 
Define $K_h$ by~\eqref{eq:def.h} and denote its  stationary distribution by $\pi_h$.  
For $k = 1, \dots, t$,   \smallskip  \\ 
\noindent (i) draw $x^{(k)}$ from $K_h(x^{(k-1)}, \cdot)$,  \smallskip  \\ 
\noindent (ii) calculate   $\comega^{(k)} \propto  \pi(x^{(k)}) / \pi_h(x^{(k)})$.  \smallskip  \\
Return samples $x^{(1)}, \dots, x^{(t)}$ and their un-normalized importance weights $\comega^{(1)}, \dots, \comega^{(t)}$. 
\end{alg}

\begin{remark}\label{rmk:iit.general}
The generic TGS algorithm was proposed as a Gibbs sampler on a product space, which updates one coordinate by conditioning on all the others. 
Algorithm~\ref{alg:iit} generalizes it to arbitrary finite spaces. 
Some algebra shows that the TGS algorithm for variable selection introduced in~\citet[Section 4.2]{zanella2019scalable} is a special case of Algorithm~\ref{alg:iit} with $h = \htgs$. 
We explain  in detail the link between TGS and Algorithm~\ref{alg:iit} in Supplement~\ref{supp:tgs}. 
\end{remark}

To implement Algorithm~\ref{alg:iit}, we need to evaluate  $\pi_h$.  The following lemma gives $\pi_h$ for two classes of reversible IIT schemes that are of particular interest to this work. 

\begin{lemma}\label{lm:pih}
Let $K_h$ be as given in~\eqref{eq:def.h} with stationary distribution $\pi_h$. 
If $h(u) = u^a$ for some $a \geq 0$, then  $\pi_h \propto \pi^{2a} Z_h$. 
If $h$ is a balancing function, then $\pi_h  \propto \pi Z_h$.  
Further, $K_h$ is reversible in both cases. 
\end{lemma} 

In Supplement~\ref{S:alg}, we detail the implementation of IIT using Lemma~\ref{lm:pih} for the two classes of function $h$ considered above. 
When $h$ is a balancing function, we refer to Algorithm~\ref{alg:iit} as a locally balanced IIT scheme. 

Suppose $h$ is non-decreasing. Since $Z_h$ is defined as the sum of  $h(\pi(y) / \pi(x))$ for all neighboring states $y$, $Z_h$ is most likely to achieve its minimum (resp. maximum) at a local maximum (resp. minimum) point of $\pi$. But in most parts of $\cX$, we expect that $Z_h(x)$ does not depend much on $\pi(x)$ and it behaves just like a random noise. 
Hence, Lemma~\ref{lm:pih} suggests that  for locally balanced IIT schemes,  the distribution $\pi_h$ can be seen as a random perturbation of $\pi$ (except around local extrema of $\pi$). 
This is an intuitive reason why locally balanced IIT schemes may work well as importance sampling tends to be most efficient when  $\pi_h$ looks similar to the target distribution $\pi$~\citep[Chapter 2.5]{liu2008monte}.   
For $h(u) = u^a$, by Lemma~\ref{lm:pih}, we have $\omega = \pi/\pi_h \propto \pi^{1 - 2a} Z_h^{-1}$. 
Ignoring the term $Z_h^{-1}$, we see that  if $a > 1/2$, states with negligible posterior probabilities can receive exceedingly large importance weights, which can cause the estimator~\eqref{eq:sni} to converge very slowly. 
This will be confirmed in the next section by analyzing a toy example. 
 
\subsection{Measuring Rates of Convergence}\label{sec:gap}
Given an irreducible and reversible transition matrix $P$,  we can  denote its eigenvalues by  $1 = \lambda_1(P) > \lambda_2(P) \geq \cdots \geq \lambda_{|\cX|}(P) \geq -1$, and define its spectral gap  by $\gap(P) = 1 - \lambda_2(P).$ 
If all eigenvalues of $P$ are non-negative, it is well known that the mixing time of $P$ can be bounded by $C_\pi \gap(P)^{-1} $, where the constant $C_\pi$ only depends on $ \min_{x \in \cX} \pi(x)$~\citep{sinclair1992improved}. 
If $Q$ is the  transition rate matrix of an irreducible and reversible continuous-time Markov chain, it has eigenvalues $0 = \lambda_1(Q) > \lambda_2(Q) \geq \cdots \geq \lambda_{|\cX|}(Q) $, and we define its spectral gap by $\gap(Q) =   - \lambda_2(Q).$

Though the performance of the IIT estimator partially depends on the mixing rate of $K_h$, $\gap(K_h)$ does not reflect the overall efficiency of IIT since it does not take into account the importance weights.  
But if we replace $\omega(x^{(k)})$ in~\eqref{eq:sni} with an exponential random variable with mean $\omega(x^{(k)})$, we obtain the (unweighted) time average of a continuous-time Markov chain. This motivates us to use the  spectral gap of this new chain, which we denote by $Q_h$,  to measure  the ``convergence rate'' of IIT. 
The importance weight of state $x$, $\omega(x)$, becomes the average holding time at state $x$ of the chain $Q_h$. 
The following result was proved in~\citet[Lemma 2]{zanella2019scalable}  for TGS by using a variational characterization of the asymptotic variance~\citep{andrieu2016establishing}.  We give a different proof in the supplement. 
As usual, $N(0, \sigma^2)$ denotes the normal distribution (or a normal random variable) with mean $0$ and variance $\sigma^2$. 
\begin{lemma}\label{lm:tgs}
Consider Algorithm~\ref{alg:iit} and assume $K_h$ is reversible with respect to $\pi_h$. 
Define a transition rate matrix $Q_h$ by 
\begin{equation}\label{eq:def.Q}
Q_h(x, y) = \left\{\begin{array}{cc}
K_h(x, y) / \omega(x), \quad & \text{ if } x \neq y, \\
- \sum_{x' \neq x} Q_h(x, x'), \quad & \text{ if } x = y,  \\
\end{array}
\right.
\end{equation}
where $\omega = \pi / \pi_h$. 
Let  $x^{(1)}, \dots, x^{(t)}$ be samples generated from $K_h$. Consider the estimator $\hat{f}(t, \omega) $ defined in~\eqref{eq:sni} for some function $f$ such that $\bbE_\pi [f] = 0$. 
Then, $\sqrt{t} \hat{f}(t, \omega)  \overset{D}{\rightarrow} N(0,   \sigma^2  )$   where $\overset{D}{\rightarrow}$ denotes the convergence in distribution and $\sigma^2 \leq 2 \bbE_\pi[f^2] / \gap(Q_h)$. 
\end{lemma}
 
\begin{remark}\label{rmk:Qh}
We call $\sigma^2$ the asymptotic variance of the estimator $\hat{f} (t, \omega)$. 
Analogously, the asymptotic variance of an unweighted time average  of a discrete-time Markov chain $P$ can be bounded by $\bbE_\pi[f^2] / \gap(P)$, regardless of the periodicity of $P$~\citep[Proposition 1]{rosenthal2003asymptotic}.  
Henceforth, we use $\gap(Q_h)$  to measure the convergence rate of Algorithm~\ref{alg:iit}, and similarly, the convergence rate of RWMH is measured by the spectral gap of its transition matrix. 
\end{remark} 

From the construction of $Q_h$, we see that $\omega(x)$ needs to be small for states $x$ with negligible posterior probabilities so that the chain $Q_h$ can quickly move to better states. Below we analyze the importance weight $\omega$ for a toy model, which provides important insights  into the choice of $a$  for IIT schemes with $h(u) = u^a$. 

\begin{example}\label{ex:bad}
Let $\cX = \{0, 1, \dots, p\}$ and $\cN(x) = \{ y \in \cX \colon |x - y| = 1 \}$.  Assume that $\pi(0) = \pi(1)$, and for $k = 1, \dots, p-1$,  $\pi(k) / \pi(k + 1) \geq r$ for some constant $r = r(p)$ such that  $r \rightarrow \infty$ as $p \rightarrow \infty$. 
This setup ensures that $\pi(k)$ quickly decays as $k$ grows so that $\pi$ concentrates on only states $0$ and $1$. 
Letting $\sim$ denote asymptotic equivalence as $p \rightarrow \infty$, one can show that $\pi(0) = \pi(1) \sim 1/ 2$. 
Consider IIT schemes with $h(u) = u^a$ for some fixed $a > 0$.  
In Supplement~\ref{supp:bad}, we show that as $p \rightarrow \infty$, 
\begin{align*}
    \omega(k) = \frac{\pi(k)} {\pi_h(k)} \sim 2^{1-2a} \frac{  \pi(k)^{1-a}  }{  \pi(k - 1)^a }, \quad   k = 1, \dots, p. 
\end{align*} 
We make a few observations.  \smallskip \\
 \noindent  (i) If $a \geq 1$,  we have $ \omega(p) \gtrsim 2^{1-a} r^{(p - 2)a}$, which grows super-exponentially with $p$ (as we assume $r \rightarrow \infty$).  \smallskip \\
 \noindent  (ii)   If $0 < a \leq 1/2$, for each $k \geq 2$, we have $\omega(k) \lesssim 2 r^{-a}$, which goes to zero as $r \rightarrow \infty$. \smallskip \\
 \noindent  (iii)    Suppose for $k = 1, \dots, p$, $\pi(k) \propto r^{ - k^{c} }$ for some universal constant $c  \geq  1$. One can show that for any $a > 1/2$, $\omega(p)$ still grows super-exponentially with $p$.  \smallskip \\
Hence, for this toy model, $Q_h$ always mixes quickly for any $a \in (0, 1/2]$, while $a > 1/2$ can easily make $Q_h$ mix slowly, even if the tail of $\pi$ decays quickly.  
Since a very small value of $a$ defeats the purpose of using a locally informed sampling scheme, this example suggests that we may want to use $ a = 1/2$ in practice. 
Interestingly,  $h(u) = u^a$ is a balancing function only when $a = 1/2$. 
\end{example}

\section{\MakeUppercase{Spectral Gap Bounds}}\label{sec:main}

\subsection{Results for Unimodal Targets} \label{sec:unimodal}
MCMC sampling from continuous distributions has been extensively studied in the literature. A commonly used assumption in these works is the log-concavity or strong log-concavity of the target; see~\citet[Section 9.10]{saumard2014log} for a brief review, and for recent works, see~\citet{dalalyan2017theoretical, mangoubi2017rapid, dwivedi2018log, cheng2018underdamped, mangoubi2019mixing, shen2019randomized}, among many others. 
All log-concave continuous distributions are unimodal and have sub-exponential tails. 
In our setting, we propose to consider the following condition on $\pi$, which also assumes unimodality and ``sub-exponential tails'' but conceptually requires less than log-concavity. 
 
\begin{assumption}\label{A1}
Let $|\cX| < \infty$, $\pi(x) > 0$ for each $x$, and $\cN$ be a symmetric neighborhood mapping such that $ \max_{x \in \cX}| \cN(x) | \leq p^\alpha$ for some $\alpha > 0$ and $p > 1$.   
There exist a state $x^* \in \cX$, an operator $\cT \colon \cX \rightarrow \cX$ and some constants $\nu > \alpha$ such that $\cT(x) \in \cN(x)$ and $\pi( \cT(x) ) \geq p^\nu \pi(x)$ for each $x \neq x^*$. Define $\cT(x^*) = x^*$. 
\end{assumption}
  
\begin{remark}\label{rmk:uni1} 
Observe that under Assumption~\ref{A1}, we have $x^*  = \argmax_{x \in \cX} \pi(x)$, and  $x^*$ is the unique fixed point of $\cT$.
To see the link between Assumption~\ref{A1} and log-concavity, first consider the one-dimensional case where $\cX = \{0, 1, \dots, m \}$ for some finite $m$. Then $\pi$ is log-concave if $\pi(x)^2 \geq \pi(x + 1) \pi(x - 1)$ for each $x$~\citep{saumard2014log}. Assuming $\pi$ is maximized at $x = 0$, we can define $\cT(x) = x - 1$ for each $x > 0$, and use log-concavity of $\pi$ to show that $\pi(\cT(x)) / \pi(x) \geq \pi(0) / \pi(1)$ for each $x \geq 1$. 
In higher dimensions, we point out that Assumption~\ref{A1} requires a much weaker notion of unimodality than log-concavity.  
If $\varphi$ is a log-concave  probability density function  defined on $\bbR^p$, then for any unit vector $\se$, the mapping $x \mapsto \varphi(x \se)$ for $x \in \bbR$ is again log-concave and thus unimodal. 
In contrast, Assumption~\ref{A1} does not require $\pi$ to be unimodal ``in every direction''; see Supplement~\ref{supp:uni} for a toy example. This is very important when we consider high-dimensional model selection; see Example~\ref{ex:var2} (continued) below. 
The condition $\nu > \alpha$ is imposed in Assumption~\ref{A1} so that $\pi$ has ``sub-exponential tails''. 
Let $D(x) = \min \{k \colon \cT^k (x) = x^*\}$ denote the distance from $x$ to the mode $x^*$.  One can show that  $D$ is sub-exponential:   $\pi(\{D \geq k\}) \leq e^{- c k}$  for $c = (\nu - \alpha) \log p$.  
\end{remark}

For some model selection problems, we can rigorously verify Assumption~\ref{A1} by imposing mild high-dimensional conditions and letting $x^*$ be the ``true'' model  (more precisely, $x^*$ is often defined as the model that contains all signals that exceed some threshold). 
Below we briefly explain why we expect it holds for variable selection, and we refer readers to~\citet{zhou2021complexity} for how to establish Assumption~\ref{A1} for sparse structure learning. 

\setcounter{example}{0}
\begin{example}[continued]\label{ex:var2}
Consider Example~\ref{ex:var1} again.   
\citet{yang2016computational} considered high-dimensional Bayesian variable selection with a standard g-prior for linear regression models and a sparsity prior on $x$. In their Lemma 4, the authors essentially proved that, under some reasonable high-dimensional conditions,  Assumption~\ref{A1} holds for the triple $(\cX, \cN^2, \pi)$ where we recall $\cN^2$ denotes the ``add-delete-swap'' neighborhood and it satisfies  $|\cN^2(x)| = O(s p)$ for any $x$.  
To explain their construction of the function $\cT$,  let $x^*$ denote the ``true'' model, and we say the $i$-th variable is ``influential'' if $x^*_i = 1$ and ``non-influential'' if $x^*_i = 0$.  We say   $x$ is ``overfitted'' if $x_i \geq x^*_i$ for each $i \in \{1, \dots, p\}$.  
First, if $x \neq x^*$ is overfitted, we  let $\cT(x)$ be the best model that we can obtain by removing a non-influential variable from  $x$  (``best'' means it maximizes $\pi$). 
Second, if $x$ is underfitted (i.e., not overfitted) and $\norm{x}_1 < s$,  we let  $\cT(x)$ be the best model that we can obtain by adding an influential variable. Finally, if $x$ is underfitted and $\norm{x}_1 = s$, we use the best swap move that exchanges a non-influential variable with an influential one. 
By construction, we have $\norm{ \cT(x) - x^*}_1 < \norm{ x - x^*}_1$ for any $x \neq x^*$. This explains why Assumption~\ref{A1} can be seen as a consistency  property of  model  selection procedures: it assumes that  any $x \neq x^*$ has a neighbor $\cT(x)$  which has a larger posterior probability and   is ``more similar'' to $x^*$ than $x$.  
It is important to note that in the proof of~\citet{yang2016computational},  $\cT(x)$ is always defined to be the best move of its type.  
In particular, if $x$ is underfitted, there is no guarantee that we can increase the posterior probability by adding any influential variable that is not in $x$, because of the collinearity of the design matrix. 
This underscores one point made in Remark~\ref{rmk:uni1}: under Assumption~\ref{A1}, $\pi$ can still look very ``irregular'' due to the dependence among coordinate variables (if $\cX$ takes a product form), which is very likely to happen in high-dimensional model selection.  
Lastly, we note that from a purely theoretical standpoint, there is often no loss of generality in assuming that $\nu$ in Assumption~\ref{A1} is a sufficiently large universal constant, which is explained in detail in~\citet[Supplement S2]{zhou2021local} for  the variable selection problem.  
\end{example}
\setcounter{example}{3}
 
It is known that when the unique mode is ``sufficiently sharp'', we can bound the spectral gap of a reversible Markov chain   using ``canonical paths''~\citep{sinclair1992improved}. 
In Lemma~\ref{lm:main} below, we further improve the existing spectral gap bounds~\citep{yang2016computational, zhou2021complexity} by using a refined path argument.  The key idea of our proof is to measure the ``length'' of each edge in the canonical paths in light of $\pi$. 
 
\begin{lemma}\label{lm:main} 
Suppose Assumption~\ref{A1} holds. 
For any transition matrix $P$ or transition rate matrix $Q$ that is irreducible and  reversible with respect to $\pi$, we have 
\begin{equation*}
\gap(P) \geq   \kappa(p, \alpha, \nu) \min_{x \neq x^*}  P(x, \cT(x)), 
\end{equation*}
where
\begin{equation*}
\kappa(p, \alpha, \nu) =  \frac{1}{2}  \left\{ 1 -   p^{-(\nu - \alpha) / 2}   \right\}^3, 
\end{equation*}   
and  $\gap(Q) \geq  \kappa(p, \alpha, \nu) \min_{x \neq x^*} Q(x, \cT(x))$. 
\end{lemma}

\begin{remark}\label{rmk:complexity}
Lemma~\ref{lm:main} is non-asymptotic.  For high-dimensional model selection problems, we can consider the asymptotic regime where $p \rightarrow \infty$ and $\nu > \alpha$ are fixed constants. 
Then,  by Lemma~\ref{lm:main}, the convergence rate of an MCMC algorithm has the same order as the minimum transition probability/rate from $x$ to $\cT(x)$. 
For RWMH, one can use Lemma~\ref{lm:main} to show that  its convergence rate is $O(p^{-\alpha})$. 
Since each IIT iteration has complexity $O(p^\alpha)$,   we need $\gap(Q_h) \geq c$ for some universal constant $c$ so that the ``real-time convergence rate'' of IIT  is at least as fast as that of RWMH. 
\end{remark}

\begin{theorem}\label{th:one}
Suppose Assumption~\ref{A1} holds and let $\kappa(p, \alpha, \nu)$ be as given in Lemma~\ref{lm:main}. 
For Algorithm~\ref{alg:iit} with some non-decreasing balancing function $h$, we have 
\begin{equation*} 
\gap(Q_h)  \geq  \kappa(p, \alpha, \nu) \frac{h(p^\nu) }{ \bbE_\pi[Z_h ] }, 
\end{equation*} 
where $Q_h$ is as defined in~\eqref{eq:def.Q},
Further,  for the three balancing functions defined in~\eqref{def:three}, we have 
\begin{align*}
\frac{ 2 \gap(Q_h) }{   \kappa(p, \alpha, \nu)  } \geq \left\{\begin{array}{ll}
p^{\nu - \alpha},     &  \text{ if } h = \htgs, \\ 
p^{\nu - 2 \alpha},         &  \text{ if } h = \hmin,  \\
p^{\nu / 2}/ (p^{2 \alpha - \nu} + p^{\alpha - \nu / 2}   ),   &  \text{ if } h  = \hsq. 
\end{array} \right.
\end{align*}
\end{theorem}
 
Theorem~\ref{th:one} provides the theoretical guarantee for the scalability of IIT schemes when the posterior mass concentrates on just one model $x^*$.  
Consider $\htgs$ for example, which was used by the TGS algorithm for variable selection proposed in~\citet[Section 4.2]{zanella2019scalable}. 
Theorem~\ref{th:one} suggests that the real-time convergence rate of TGS  is $O(p^{\nu - 2 \alpha})$ under Assumption~\ref{A1}, which is always faster than that of RWMH by Remark~\ref{rmk:complexity}. 
If $\nu > 2 \alpha $ (which happens when the sample size is sufficiently large),   $\hsq$ is as efficient as $\htgs$. 
  
One may find Theorem~\ref{th:one} to be counter-intuitive.  If $\nu > 2 \alpha$, our bound implies that the IIT scheme $\htgs$ converges even faster (in real time) for larger $p$.  
To gain insights into this surprising phenomenon, recall that for locally balanced IIT schemes the importance weight $ \pi / \pi_h \propto Z_h^{-1}$. For $h = \htgs$, under Assumption~\ref{A1}, $Z_h(x) \geq p^\nu$ for any $x \neq x^*$ while $Z_h(x^*) \approx p^\alpha$; that is, $x^*$ receives a much larger importance weight than any other state.  Consequently, the estimator $\hat{f}(t, \pi / \pi_h)$  defined in~\eqref{eq:sni} becomes almost a constant (i.e., variance gets reduced to almost zero) once the chain visits $x^*$, which should happen quickly due to the use of informed proposals. 
Essentially, under Assumption~\ref{A1}, when $\nu$ is sufficiently large, an IIT sampler can behave just like a single best model selection procedure due to the importance weighting. 

\subsection{Results for Targets Concentrating on a Set}\label{sec:tgs.two}
The excellent scalability of IIT demonstrated by Theorem~\ref{th:one}  needs to be taken with a grain of salt. 
In practice, even if the sample size is very large, the posterior distribution may fail to satisfy the unimodal condition in Assumption~\ref{A1} due to collinearity among the variables  or identifiability issues; for example, in structure learning, observational data cannot distinguish between Markov equivalent directed acyclic graphs. 
In such cases, we expect that there exists a small set $\cX^* \subset \cX$ such that $\pi(\cX^*)$ is close to one and for any $x \notin \cX^*$, an MCMC sampler can easily move from $x$ to $\cX^*$. However, because now IIT has to visit all those models in $\cX^*$ many times to reduce the variance of the estimator, importance weighting can no longer boost the sampler's efficiency as significantly as under Assumption~\ref{A1}.  
Consider the following setting which generalizes Assumption~\ref{A1} to cases where $\pi$ concentrates on the set $\cX^*$.  
\begin{assumption}\label{A2}
Let $\cX, \pi, \cN$ be as given in Assumption~\ref{A1} such that $\max_{x \in \cX} |\cN(x)| \leq p^\alpha$ for some $\alpha > 0$, $p > 1$. 
There exist a set $\cX^*  \subset \cX$ with $|\cX^*| \geq 2$, an operator $\cT \colon \cX \rightarrow \cX$ and  constants $\nu > \alpha$ and $B \geq 1$ such that (i)  $\max_{x, x' \in \cX^*} \pi(x') / \pi(x) \leq B$, (ii) $\cT(x) \in \cN(x)$ for each $x$, (iii) $\pi( \cT(x) ) \geq p^\nu \pi(x)$ for each $x \notin \cX^*$. 
\end{assumption}

\begin{remark}
Condition (i) means all states in $\cX^*$ have comparable stationary probabilities. 
Conditions (ii) and (iii), roughly speaking, imply that $\pi$ satisfies Assumption~\ref{A1} if we collapse all states in $\cX^*$ into one single state. 
\end{remark} 

Before we state the spectral gap bounds for Assumption~\ref{A2}, we define ``restricted'' Markov chains. 

\begin{definition}\label{def:restrict}
Given a transition matrix $K$ with state space $\cX$, its restriction to $S \subset \cX$ is a transition matrix $K_S \colon S^2 \rightarrow [0, 1]$  such that $K_S(x, x') = K(x, x')$  if $x \neq x'$,  and $K_S(x, x) = 1 - \sum_{x' \in S \setminus \{x\}} K (x, x')$. 
Similarly, for a transition rate matrix $Q$, its restriction to $S$ is a mapping $Q_S \colon S^2 \rightarrow \bbR$ such that $Q_S(x, x') = Q(x, x')$  if $x \neq x'$,  and $Q_S(x, x) =   - \sum_{x' \in S \setminus \{x\}} Q (x, x')$. 
\end{definition}

\begin{theorem}\label{th:two}
Suppose Assumption~\ref{A2} holds with $B = 1$.
Let $ \kappa(p, \alpha, \nu)$ be as given in Lemma~\ref{lm:main} and $M = |\cX^*|$. 
Consider Algorithm~\ref{alg:iit} with some non-decreasing balancing function $h$. 
If $K_h$ restricted to $\cX^*$ is irreducible, 
\begin{align*}
          \gap(Q_h) \geq \frac{\kappa(p, \alpha, \nu) h(1) }{3 (Mp^{\alpha - \nu} +1) M (M-1) \bbE_\pi[Z_h]  }.
\end{align*}
Further,  for the three balancing functions defined in~\eqref{def:three}, we have $h(1) \geq 1$ and 
\begin{align*}
\bbE_\pi[Z_h] \leq \left\{\begin{array}{ll}
2p^\alpha,   &  \text{ if } h = \htgs, \\ 
2 p^{2\alpha - \nu} + M - 1,         &  \text{ if } h = \hmin,  \\
p^{2\alpha - \nu} + 2p^{\alpha - \nu/2} + M - 1,    &  \text{ if } h = \hsq. 
\end{array} \right.
\end{align*}
\end{theorem} 

\begin{remark}
The condition $B = 1$ is used merely for ease of presentation, and it will be removed later in Theorem~\ref{th:mix}.   
In our proof of Theorem~\ref{th:two}, we use a general ``worst-case'' bound on the spectral gap of $Q_h$ restricted to  $\cX^*$, which yields the $O(M^2)$ term on the denominator of the bound.
It may be significantly improved when one applies Theorem~\ref{th:two} to a specific problem. 
\end{remark}

Assuming $M$ is bounded as $p \rightarrow \infty$,  by Theorem~\ref{th:two}, we have $\gap(Q_h)^{-1} = O( \bbE_\pi[Z_h])$, from which we see that $\htgs$ appears to be the least efficient, while we recall $\htgs$ is the best in Theorem~\ref{th:one}. 
Indeed, our estimate suggests that for $h = \htgs$, $\gap(Q_h)$ may have the same order as the spectral gap of RWMH (so RWMH would perform much better in reality since each iteration of RWMH runs much faster). 
For the other two choices $\hmin$ and $\hsq$, we have $\gap(Q_h) \gtrsim c$ for some universal constant $c$ if Assumption~\ref{A2} holds with $\nu - 2 \alpha > \epsilon$ for some fixed $\epsilon > 0$. 
Hence, when the posterior mass concentrates on more than one models, the IIT schemes $\hmin$ and $\hsq$ tend to have better scalability than the scheme $\htgs$. 
We use variable selection as an example to show that for $h = \htgs$, the bound $\gap(Q_h)^{-1} = O(p^{-\alpha})$ can be attained. 

\begin{example}\label{ex:counter} 
Let $\cX = \{0, 1\}^p$ and, as in  Example~\ref{ex:var1}, let $\cN^1(x) = \{y \in \cX \colon \norm{x - y}_1 = 1 \}$. Assume  $\pi(x) \propto r^{ - d(x)}$ for some $r \geq 1$ and $d(x) = \sum_{i = 2}^p x_i$. Hence, the empty model and the model with only variable $1$ are equally the best. 
 Consider Algorithm~\ref{alg:iit} with $h(u) = 1 + u$, and let $Q_h$ be as given in~\eqref{eq:def.Q}.   It can be shown that $\gap(Q_h) \leq 2 / p$; see Supplement~\ref{supp:counter}. 
\end{example}

We can also derive a mixing time bound under Assumption~\ref{A2}. To this end, we first prove a generic decomposition bound in  Supplement~\ref{supp:mix}, similar to those found in~\citet{pillai2017elementary}, by studying the ``trace'' of a Markov chain. This approach allows us not to require the irreducibility of $K_h$ restricted to $\cX^*$ as in Theorem~\ref{th:two}. 
 
\begin{definition}\label{def:trace}
Fix a discrete-time Markov chain $\{X_{t}\}_{t \geq 0}$ and a subset $S \subset \mathcal{X}$ of its state space. Define $\eta_{0} = \min\{t \colon  X_{t} \in S\}$, then inductively set:
\begin{align*}
\eta_{i+1} = \min \{t > \eta_{i} \, : \, X_{t} \in S\}.
\end{align*}
The \textit{trace} of $\{X_{t}\}_{t \geq 0}$ on $S$ is the process $\{X_{\eta_{t}}\}_{t \geq 0}$. If $K$ is the transition matrix of $\{X_{t}\}_{t \geq 0},$ define $K|_{S}$ to be the transition matrix of the trace. 
\end{definition}

\begin{theorem}\label{th:mix}
Suppose Assumption~\ref{A2} holds with $p^{\nu - \alpha}$ being sufficiently large. 
Consider Algorithm~\ref{alg:iit} with some non-decreasing balancing function $h$. 
Let $b = 2 \max_{ x \in \cX} |Q(x, x)|$ and $P = b^{-1} Q + I$. 
If there exists $\delta > 0$ such that the graph with vertex set $\cX^*$ and edge set $\{ (x, x') \in \cX^2 \colon P |_{\cX^*}(x, x') \geq \delta / b \}$ is connected, 
\begin{align*}
     \gap(Q_h) \geq \tmix(Q_h)^{-1}  
     \geq  C \min  \Big\{
    \frac{\delta}{B M^2 \log (4 B M)},   
     \frac{ \kappa \, h (p^\nu)}{M    \bbE_\pi[Z_h] \log(8 M / \pmin )}
    \Big\}, 
\end{align*} 
where $C>0$ is a universal constant, $ \kappa$ is as given in Lemma~\ref{lm:main}, $M = |\cX^*|$,  $\pmin = \min_{x \in \cX} \pi(x)$, and 
\begin{align*}
    \tmix(Q_h) = \max_{x} \, \min\{t \colon 
 \norm{   e^{t Q_h}(x, \cdot) - \pi(\cdot) }_{\rm{TV}} \leq 1/4 \} 
\end{align*}
denotes the mixing time of $Q_h$ ($\mathrm{TV}$ stands for total variation distance). 
\end{theorem}

For most high-dimensional model selection problems, we may assume that $\log(\pmin^{-1})$ grows at most polynomially in $p$.  
Thus, the bounds on $\bbE_\pi[Z_h]$ provided in Theorem~\ref{th:two} suggest that, if $\nu$ is sufficiently large and $\delta, B, M$ are bounded from above, we have $\tmix(Q_h) = O(\delta^{-1})$.  
The advantage of Theorem~\ref{th:mix} is that the trace chain $P|_{\cX^*}$
 mixes much more quickly than the restriction chain used in Theorem \ref{th:two}.
As an extreme but common illustration, the trace is always ergodic and mixes at least as quickly as $P$, while the restriction chain may fail to be ergodic (e.g. if $\cX^*$ is not connected by moves of the chain). 

\begin{example}\label{ex:ex3}
Let $\cX = \{0, 1\}^p$ and $\cN^1$ be as given in Example~\ref{ex:var1}. 
Let $x^i$ denote the model with only the $i$-th variable, e.g., $x^2 = (0, 1, 0, 0, \dots, 0)$. 
Assume that $x^1, x^2$ are two equally best models, and let $\cX^* = \{x^1,  x^2\}$.  
We can not apply Theorem~\ref{th:two} since $x^1$ and $x^2$ are not neighbors ($\cN^1$ does not include swap moves).  
But Theorem~\ref{th:mix} can be applied, and  for the transition matrix $P$ defined in the theorem, we have $$P|_{\cX^*}(x^1, x^2) \geq   P(x^1, x^0) P(x^0, x^2) / (1 - P(x^0, x^0) )$$ 
where $x^0$ denotes the empty model. This can be used to estimate the parameter $\delta$.   
\end{example}

\subsection{Results for Decomposable Targets} \label{sec:decomp}
We can further relax Assumption~\ref{A2} using decomposable Markov chains. 
Given a surjective mapping $\cG \colon \cX \rightarrow \cY$ where $\cY$ is a finite space with $|\cY| < |\cX|$, we let $\cX_y = \{x \in \cX \colon \cG(x) = y\}$ and 
$\pi_\cG$ be the push-forward measure of $\pi$ on $\cY$; that is, $\pi_{\cG}(y) = \pi(\cX_y)$   for each $y \in \cY$. 

\begin{assumption}\label{A3}
Let $|\cX| < \infty$, $\pi(x) > 0$ for each $x \in \cX$ and $\cN$ be a symmetric neighborhood function on $\cX$. 
Let $|\cY| < |\cX|$, $\cG \colon \cX \rightarrow \cY$ be  surjective, and $\cT \colon \cY \rightarrow \cY$ be an operator with a unique fixed point $y^* \in \cY$. 
For each $y \neq y^*$, define 
\begin{align*}
    \cX_y( \cT(y), c) =\left\{ x \in \cX_y \colon   \max_{x' \in \cX_{\cT(y)}  \cap \cN(x) } \pi(x')  \geq p^c \pi(x) \right\}. 
\end{align*}
Suppose there exist universal constants $\nu > \alpha > 0$, $\varepsilon > 0$ and $\tilde{\nu} > 0$ such that 
     (i) for each $y\in \cY$, $|\{y' \in \cY \colon \cT (y') = y  \} | \leq p^\alpha$,  (ii) for each $y \neq y^*$, $\pi_\cG( \cT(y) ) \geq p^\nu \pi_\cG(y)$  and  (iii) for each $y \neq y^*$, $\pi( \cX_y(\cT(y), \tilde{\nu}) ) \geq \varepsilon \pi(\cX_y)$. 
\end{assumption}

\begin{remark}
In Assumption~\ref{A3}, the mapping $\cG$ can be interpreted as a clustering of the states: for each $y$, states in $\cX_y$ tend to be close to each other and have similar posterior probabilities. Conditions (i) and (ii) essentially assume that $\pi_\cG$ satisfies a unimodal condition on $\cY$. 
The set $ \cX_y(\cT(y), \tilde{\nu}) $   can be interpreted as a collection of states from which a locally informed Markov chain is likely to enter the set $\cX_{\cT(y)}$, 
and condition (iii) ensures that the stationary measure of $ \cX_y(\cT(y), \tilde{\nu}) $ is not too small. 
Compared with Assumption~\ref{A2}, here we allow $\pi$ to be multimodal even outside of the set $\cX_{y^*}$ where the posterior concentrates. 
\end{remark}

\begin{theorem}\label{th:decomp} 
Suppose Assumption~\ref{A3} holds. Consider Algorithm~\ref{alg:iit} with some non-decreasing balancing function $h$, and let $Q_h$ be as defined  in~\eqref{eq:def.Q}. 
Let $Q_y$ be the restriction of $Q_h$ to $\cX_y$.  We have 
\begin{align*} 
\gap(Q_h)  \geq      \kappa   \varepsilon h(p^{\tilde{\nu}})  \min  \Big\{ \frac{ 1 }{ 3 \bbE_\pi [Z_h] },  
  \frac{  \min_{y \in \cY} \gap(Q_y)}{\kappa  \varepsilon  h(p^{\tilde{\nu}}) +  3 \max_{x \in \cX} Z_h(x)} \Big\}, 
\end{align*} 
where  $\kappa = \kappa(p, \alpha, \nu)$ is as given in Lemma~\ref{lm:main}.  
\end{theorem}

\begin{remark}
The proofs of Theorems~\ref{th:two} and~\ref{th:decomp} rely on the state decomposition result of~\citet[Theorem 1]{jerrum2004elementary}. 
While there exist various spectral gap bounds for decomposable Markov chains~\citep{martin2000sampling, madras2002markov,guan2007small}, for our problem we need to be careful with which one to apply.
Most bounds estimate  the spectral gap of a Markov chain by a product of two spectral gaps, one corresponding to the mixing between components and the other corresponding to the mixing within each component. 
But in our analysis  we need to use uniformization argument to bound $\gap(Q_h)$, and such multiplicative bounds can yield poor estimates because the uniformization constants cannot cancel out. 
\end{remark}

\begin{remark}\label{rmk:bounded}
Theorem~\ref{th:decomp} is most useful when we can bound $\max_x Z_h(x)$. But this is not a limitation since one can always construct a balancing function bounded by $p^c$ for any $c > 0$.  
By Remark~\ref{rmk:balance}, one way to achieve this is to define $g(u) = u \wedge p^{c}$, and  let $h(u) = g(u) \vee u g(u^{-1}) = (u \wedge p^{c}) \vee (1 \wedge up^{c} )$. Then $h$ is a balancing function taking value in $(0, p^{c}]$.  
Intuitively, the use of a bounded $h$ should be beneficial since it prevents the informed proposal from being too ``aggressive'', which means to overly favor those neighboring states with larger posterior probabilities. 
\end{remark}

\section{\MakeUppercase{Simulation Studies}}\label{sec:sim}  
\subsection{Weighted Permutations}\label{sec:sim.perm}
In the first numerical example, we consider the ``weighted permutations'' problem~\citep{zanella2020informed} where $\cX$ is the collection of all possible permutations of $\{1, \dots, p\}$. 
For applications in statistics, we can think of $\pi$ on $\cX$ as an approximate representation of the target posterior distribution in order-based MCMC methods for structure learning~\citep{friedman2003being, agrawal2018minimal}.   
We choose $p = 100$ and simulate $\pi$ under four different settings such that $\pi$ always has one unique mode at $\tau^* = (1, \dots, p)$.  We consider the following five choices for $h$ in Algorithm~\ref{alg:iit}: $\htgs, \hmin, \hsq$, $h_{0.3}(u) = u^{0.3}$ and $h_{0.4}(u) = u^{0.4}$. We find that the IIT scheme $\hsq$ appears to be the best and the only one that is consistently better than RWMH.
As is consistent with our theory,  the  scheme $\htgs$ is most sensitive to the unimodality of $\pi$. It is  worse than RWMH when the posterior mass does not concentrate on $\tau^*$ alone. 
Details are given in Supplement~\ref{supp:sim.perm}.

\subsection{Variable Selection}\label{sec:sim.var}
In the second simulation study, we consider variable selection with  $p = 5,000$ and sample size $n = 1,000$.  We simulate the data using correlated design matrices, $20$ ``causal'' predictors and different values of signal-to-noise ratio (SNR).  In both strong and weak SNR cases, $\pi$ tends to concentrate on $1$ or $2$ models, while  in the intermediate SNR case,  $\pi$ tends to be multimodal. 
Details of the simulation settings and results are presented in Supplement~\ref{supp:sim.var}.  The code for IIT sampling is written in  \texttt{C++}.  We  summarize our main findings here. 

First, as in Section~\ref{sec:sim.perm}, the IIT scheme  $\htgs$ performs poorly and the scheme $\hsq$ overperforms RWMH in all cases. When $\pi$ is multimodal,  IIT exhibits a  huge advantage. We suspect one main reason is that IIT can get out of local modes more easily since it is ``rejection-free''.  
A more careful comparison between $\hsq$ and $h_{0.3}$ shows that $h_{0.3}$ is more robust in the intermediate SNR (multimodal) case  but less efficient in the other cases, which is expected from the theory. 
The IIT scheme $\hmin$ is also much better than RWMH but seems to be too conservative and less desirable than $\hsq$ or $h_{0.3}$. 
Second, for every sampled $x$ in IIT,  we can check  $\max_{y \in \cN(x)} \pi(y) / \pi(x)$, a key quantity in the assumptions used in our theory. We find that when the SNR is either strong or weak,  Assumption~\ref{A1} is likely to be satisfied, which is   consistent with the main result of~\cite{yang2016computational}. In the intermediate SNR case,  the ratio $\max_{y \in \cN(x)} \pi(y) / \pi(x)$ is still very large for most $x$, and thus Assumption~\ref{A2} appears to be satisfied. 

\section{\MakeUppercase{Discussion}}\label{sec:disc}
\subsection{Other Related Works} \label{sec:literature}
\citet{zanella2020informed}  proposed to use balancing functions to devise locally informed MH algorithms on discrete spaces, which we note is a generalization of the reduced-rejection-rate MH of~\citet{baldassi2017method}.  
One main advantage of IIT  is that by construction, it never gets stuck at a single state.  In contrast, the performance of an informed  MH sampler largely depends on the acceptance rate, which can be very difficult to control; the same concern applies to other importance sampling schemes built on MH chains~\citep{geyer1995annealing, rudolf2020metropolis, schuster2020markov}. In particular, it was shown in~\citet{zhou2021local} that even if Assumption~\ref{A1} holds, an informed MH algorithm  that uses proposal $K_h$ with $h = \hsq$ or $h = \htgs$ can be slowly mixing, while our Theorem~\ref{th:one} shows that the convergence of the corresponding IIT samplers can be very quick. 

Recall that  we  measure  the convergence rate of IIT using the spectral gap of a continuous-time Markov chain denoted by $Q_h$. One can directly simulate $Q_h$ and use the time average to estimate $\bbE_\pi[f]$ for any function $f$ of interest. 
This approach dates back to kinetic Monte Carlo~\citep{bortz1975new, dall2001faster} and was taken in a recent paper of~\citet{power2019accelerated} on non-reversible MCMC algorithms. The authors only considered balancing functions and called  $Q_h$ the ``Zanella process''.  
Compared with IIT,  the Zanella process essentially replaces the importance weight of state $x$  with the random holding time at $x$, which appears to be slightly less efficient.   
Our work partially answers the question raised in~\citet[Section 2.1]{power2019accelerated} on how to choose a balancing function for the Zanella process.  For the use of  balancing functions on continuous state spaces, we refer readers to~\citet{livingstone2019barker} (note that the problem becomes very different due to the availability of gradient information).  

\subsection{Using IIT in Practice} \label{sec:practical}
Observe that if $\pi$ is the uniform distribution on $\cX$, no matter what $h$ is used, $K_h$ is identical to the random walk proposal; thus, RWMH should be used since it has a much smaller time complexity per iteration.  Roughly speaking, informed samplers tend to gain an advantage over RWMH when the posterior mass concentrates on a small set of states, e.g. in  model selection problems with sufficiently large sample sizes.  
Regarding the choice for $h$, heuristically, we want $h$ to ``aggressively'' push to the mode for unimodal targets, and allow exploration for multimodal targets.  
One can try to run a long burn-in period to empirically check the multimodality, though in general measuring multimodality is difficult.  
One surprising and interesting consequence of our results is that it is usually not useful to be more aggressive than $h_{0.5}$, even in the unimodal case, making it a reasonable default if one does not wish to burn-in. 
We offer two  explanations. First, consider a state $x$ and some $y \in \cN(x)$ such that $\pi(y) \gg \max_{x' \in \cN(x) \setminus \{y\}} \pi(x')$.  Then, neither $\htgs$ nor $\hmin$ can ensure that  the proposal probability $K_h(x, y)$ is always sufficiently large, while $\hsq$ does.  
Second, among the class of informed proposals with $h(u) = u^a$, the choice $a = 0.5$ tends to yield the most efficient importance weighting.  The choice $a > 0.5$ is too aggressive, as shown in Example~\ref{ex:bad} (in the extreme case $a \rightarrow \infty$, the chain becomes deterministic).  We tried $h(u) = u^{0.6}$ in the simulation study, but the result was too poor and thus not presented.  
If $a < 0.5$, the informed proposal explores posterior modes less efficiently than that with $a = 0.5$ (in the extreme case $a  = 0$, $K_h$ becomes the random walk proposal).   

In many real problems, $\pi$ is severely multimodal. In that setting,  at least IIT can still be used to efficiently explore the local posterior landscape, and one can combine IIT with other multimodal sampling techniques. For example, one may use tempering to realize ``long-range jumps'' between modes, or partition the state space so that $\pi$ is roughly unimodal on each subspace~\citep{basse2016parallel}. 

In both of our simulation studies,  the wall time usage of each IIT iteration is much smaller than that of $|\cN(\cdot)|$ RWMH iterations, probably because the empirical complexity depends on many factors, including problem dimension, likelihood complexity, implementation, hardware, etc.  This seems to favor IIT, though we expect when the problem dimension becomes extremely large, it may be helpful to approximate $Z_h$ by only evaluating a subset of neighboring states.   
When one has enough parallel computing resources, informed MCMC samplers are usually more appealing than RWMH since the time complexity of each informed iteration can be greatly reduced by parallelizing the evaluation of $\pi$ for neighboring states.  

\subsection{Concluding Remarks}  
The convergence theory of MCMC algorithms on high-dimensional discrete state spaces is largely underdeveloped. 
One major contribution of this work is an effective and general approach to obtaining strong bounds on the variance of MCMC-based methods. The path method we use is easy to generalize and applies to many problems, and by adjusting it to the context we obtain sharp estimates for the convergence rates of IIT samplers. Overall, our theory advocates the use of IIT for model selection problems if there is sufficient data. 

\section*{Acknowledgements}
We thank the anonymous reviewers whose suggestions have helped us greatly improve the paper. 

\bibliographystyle{plainnat}
\bibliography{ref}


\clearpage
\appendix

\thispagestyle{empty}

\begin{center}
{\LARGE Supplementary Material: \TITLE} 
\end{center}


The supplement is structured as follows. 
\ref{S:alg}:  pseudocode for the two classes of IIT samplers considered in the paper. \ref{supp:tgs}: a brief review of the TGS sampler and its weighted version proposed by~\citet{zanella2019scalable}. 
\ref{supp:uni}: a numerical example which shows that under Assumption~\ref{A1}, $\pi$ is not necessarily unimodal ``in every direction''. 
\ref{supp:sim}: results and simulation details for the two numerical studies presented in the main text. 
\ref{supp:exs}: details about Examples~\ref{ex:bad} and~\ref{ex:counter} in the main text. 
\ref{supp:mix}: a generic mixing time bound for decomposable Markov chains based on  trace chains. 
\ref{supp:all.proofs}: proofs of all the theoretical results stated in the main text. 
\textbf{Simulation code} can be downloaded at \url{https://github.com/zhouquan34/IIT}.  

\section{Two IIT Algorithms}\label{S:alg}
We still use the notation defined in Section~\ref{sec:notation}. 
To avoid ambiguity, let $\cpi$ be a measure on $\cX$ such that $\cpi(x) = C \pi(x)$ for some  constant $C > 0$.  To implement IIT schemes,  we only need  $\cpi$ instead of $\pi$. 
Similarly,  we use $\check{\omega}$ to denote the un-normalized version of the importance weight. 
For each $a \geq 0$, define a function $h_a \colon \hbbR \rightarrow \hbbR$ by 
\begin{align*}
h_a(u) = u^a,   \quad \forall \, u > 0. 
\end{align*}
Algorithm~\ref{alg1} shows how to implement IIT with $ h =h_a$,  and Algorithm~\ref{alg2} is for IIT schemes with an arbitrary balancing function $h$. 

\setcounter{algocf}{1}
\begin{algorithm}\setstretch{1.3}
\caption{Informed importance tempering with $h_a$}\label{alg1}
\SetKwInput{KwInput}{Input}                 
\SetKwInput{KwOutput}{Output}  
\SetKwFunction{FnH}{h}
\KwInput{Sample space $\cX$, symmetric neighborhood function $\cN \colon \cX \rightarrow 2^\cX$, 
measure $\cpi > 0$,  constant $a \geq  0$,  number of iterations $t$, and  state $x^{(0)} \in \cX$} 
Calculate $ \cpi(y)$ for each $y \in \cN(x^{(0)})$\\ 
\For{$k=1, \dots, t$}{
     Draw $x$ from $\cN(x^{(k - 1)})$ with probability proportional to $\cpi(x)^a$  \\
     Calculate $ \cpi(y)$ for each $y \in \cN(x)$ \\
     $z^{(k)} \gets \sum_{y \in \cN(x)}  \frac{\cpi(y)^a}{ \cpi(x)^a } $ \\ 
     $\comega^{(k)} \gets  \frac{ \cpi(x)^{1-2a}}{ z^{(k)} }$ \\
     $x^{(k)} \gets x$  
}
\KwOutput{samples $x^{(1)}, \dots, x^{(t)}$ and their un-normalized importance weights 
$\comega^{(1)}, \dots, \comega^{(t)}$}
\end{algorithm}

\begin{algorithm}\setstretch{1.3}
\caption{Locally balanced informed importance tempering}\label{alg2}
\SetKwInput{KwInput}{Input}                 
\SetKwInput{KwOutput}{Output}  
\SetKwFunction{FnH}{h}
\KwInput{Sample space $\cX$, symmetric neighborhood function $\cN \colon \cX \rightarrow 2^\cX$, 
measure $\cpi > 0$, balancing function $h \colon  \hbbR \rightarrow \hbbR$,  number of iterations $t$, and  state $x^{(0)} \in \cX$} 
Calculate $h\big(  \frac{ \cpi(y) }{ \cpi(x^{(0)} ) } \big)$ for each $y \in \cN(x^{(0)})$\\ 
\For{$k=1, \dots, t$}{
     Draw $x$ from $\cN(x^{(k - 1)})$ with probability proportional to $h \big(  \frac{ \cpi(x) }{ \cpi(x^{(k-1)}) } \big)$  \\
     Calculate $h\big(  \frac{ \cpi(y) }{ \cpi(x ) } \big)$ for each $y \in \cN(x)$ \\
     $z^{(k)} \gets \sum_{y \in \cN(x)} h\big(  \frac{ \cpi(y) }{ \cpi(x ) } \big)$ \\ 
     $\comega^{(k)} \gets  \frac{1}{z^{(k)}}$ \\
     $x^{(k)} \gets x$  
}
\KwOutput{samples $x^{(1)}, \dots, x^{(t)}$ and their un-normalized importance weights 
$\comega^{(1)}, \dots, \comega^{(t)}$}
\end{algorithm}

\section{Tempered Gibbs Samplers } \label{supp:tgs}
We first review the generic TGS algorithm introduced in~\citet[Section 2]{zanella2019scalable}. 
Let $\cX = \cX_1 \times \cdots \times \cX_p$ be a product space. For convenience, we still assume $|\cX| < \infty$. 
Define  $\bcN(x) = \bigcup_{i=1}^p \bcN_i(x)$ where
$\bcN_i(x) = \{ y \in \cX \colon     y_j = x_j, \; \forall \, j \neq i   \}.$
That is, $\bcN_i(x)$ is the set of all states which can only differ from $x$ at the $i$-th coordinate; note that   this definition is slightly different from the setting considered in the main text in that we assume $x \in \bcN_i(x)$ (indicated by the overbar). 
Let $\pi(x_i \mid x_{-i})$ denote the conditional density of the $i$-th coordinate given $x_{-i} = (x_1, \dots, x_{i-1}, x_{i+1}, \dots, x_p)$. 
For each  $i$ and each possible value of $x_{-i}$, let $g(x_i \mid x_{-i})$ denote a conditional ``proposal'' distribution for the $i$-th coordinate with support $\bcN_i(x)$. 
TGS is a Markov chain with transition matrix $K_{\TGS}(x,  y)  = 0$ if $y \notin \bcN(x)$, and 
\begin{align*}
K_{\TGS}(x,  y)  \propto  \frac{g(x_i \mid x_{-i}) g(y_i \mid x_{-i})}{\pi(x_i \mid x_{-i})}  ,   \quad \forall \, i \in \{1, \dots, p\}, \,  y \in \bcN_i(x). 
\end{align*}
In~\citet[Section 3.6]{zanella2019scalable}, TGS was further generalized by introducing coordinate weight functions $\xi_i(x_{-i})$ for each $(i, x_{-i})$. 
The transition matrix of the weighted TGS scheme is given by 
\begin{align*}
K_{\WTGS}(x,  y)  \propto  \xi_i(x_{-i}) \frac{g(x_i \mid x_{-i}) g(y_i \mid x_{-i})}{\pi(x_i \mid x_{-i})} ,   \quad \forall  \,  i \in \{1, \dots, p\}, \,  y \in \bcN_i(x). 
\end{align*} 
Denote the right-hand side of the above equation by $H(x, y)$; that is, we write 
$$K_{\WTGS}(x, \cdot ) =  \frac{ H(x, \cdot ) }{  Z (x) } \ind_{\bcN(x)} ( \cdot ),  \quad 
\text{ where } Z (x) = \sum_{y \in \bcN(x)} H(x, y)  
    = \sum_{i=1}^p   \xi_i(x_{-i}) \frac{g(x_i \mid x_{-i})  }{\pi(x_i \mid x_{-i})}.$$
This can be seen as a generalization of the proposal scheme $K_h$ defined in~\eqref{eq:def.h}. Since for any $y \in \bcN_i(x)$,  
\begin{align*}
    \pi(x) H(x, y) =  \pi(x_{-i})  \xi_i(x_{-i}) g(x_i \mid x_{-i}) g(y_i \mid x_{-i}) =  \pi(y) H(y, x), 
\end{align*}
$K_{\WTGS}$ is reversible w.r.t. $\pi Z $. Further, $Z^{-1}(x)$ is the importance weight associated with $x$. 
 
Next,  we describe the TGS algorithm for variable selection proposed in~\citet[Section 4.2]{zanella2019scalable}, which turns out to be a special case of Algorithm~\ref{alg:iit}.  Let $\cX = \{0, 1\}^p$ and  define 
\begin{equation}\label{eq:yi}
y^{(i)}(x) =  (x_1, \dots, x_{i-1}, 1 - x_i,  x_{i+1}, \dots, x_p), \quad \quad i = 1, \dots, p. 
\end{equation}
Suppose that the proposal $g(x_i \mid x_{-i})$ can be written as $\tg (  \pi(x_i \mid x_{-i})  )$ for some function $\tg \colon [0, 1] \rightarrow [0, 1]$, which implies $\tg (x) + \tg (1 - x) = 1$ for any $x \in [0, 1]$. Consider a balancing function 
$$h(u)  = (1 + u) \tg ( 1/(1 + u) ) \tg(u / (1 + u)),$$ 
and define $K_h(x, \cdot) \propto h ( \pi(\cdot) / \pi(x) ) \ind_{\bcN(x)}(y)$ as in~\eqref{eq:def.h}.  
Observe that $\pi(x_i \mid x_{-i}) =   \pi(x) / [\pi(x) + \pi(y^{(i)}(x))]$. 
Hence, for $y = y^{(i)}(x)$, we have 
\begin{align*}
K_h(x, y) \propto \;&  \left( 1 +  \frac{ \pi(y)}{ \pi(x) } \right) \tg \left(  \frac{ \pi(x) }{ \pi(x)+ \pi(y)  }  \right) \tg \left(  \frac{ \pi(y)  }{ \pi(x)+ \pi(y)} \right)  
=   \frac{g(x_i \mid x_{-i}) g(y_i \mid x_{-i})}{\pi(x_i \mid x_{-i})},  
\end{align*} 
which is a TGS scheme. \citet[Section 4.2]{zanella2019scalable} used the above transition matrix $K_h$ with  $\tg \equiv 1/2$; that is, it is an IIT scheme with $h(u) = 1 + u$. 
\citet[Section 4.2]{zanella2019scalable} further pointed out that, by using the argument of~\citet[see also]{liu1996peskun}, one can replace the neighborhood $\bcN(x)$   by $\cN^1(x) = \bcN(x) \setminus \{x\} =  \{ y^{(1)}(x), \dots, y^{(p)}(x)\}$. 
In our formulation of the IIT sampler, this replacement is clearly valid since we only require the neighborhood relation to be symmetric and ``connect'' all states in $\cX$. 

\section{On the Unimodal Condition in Assumption~\ref{A1}}  \label{supp:uni}
In the left panel of Figure~\ref{fig:uni}, we construct a discrete unimodal distribution on $\cX = \{1, 2, 3\}^2$. For each $x \in \cX$, its neighborhood is defined by $\cN(x) = \{y \in \cX \colon \norm{x - y}_1 = 1 \}$ where $\norm{\cdot}_1$ denotes the $L^1$-norm. The distribution $\pi$ has only one local mode (which is also the global mode) at $(1, 3)$. For example, $\pi$ is monotone increasing  along the path indicated by the black arrows in Figure~\ref{fig:uni}.  
However, the conditional distributions $\pi(1, \cdot)$ and $\pi(2, \cdot)$ are not unimodal. If we ``extend'' this distribution to a continuous state space, as shown in the right panel of Figure~\ref{fig:uni}, we get a unimodal distribution which is not log-concave. 

\begin{figure}[htp!]
    \centering
    \includegraphics[width=0.4\linewidth]{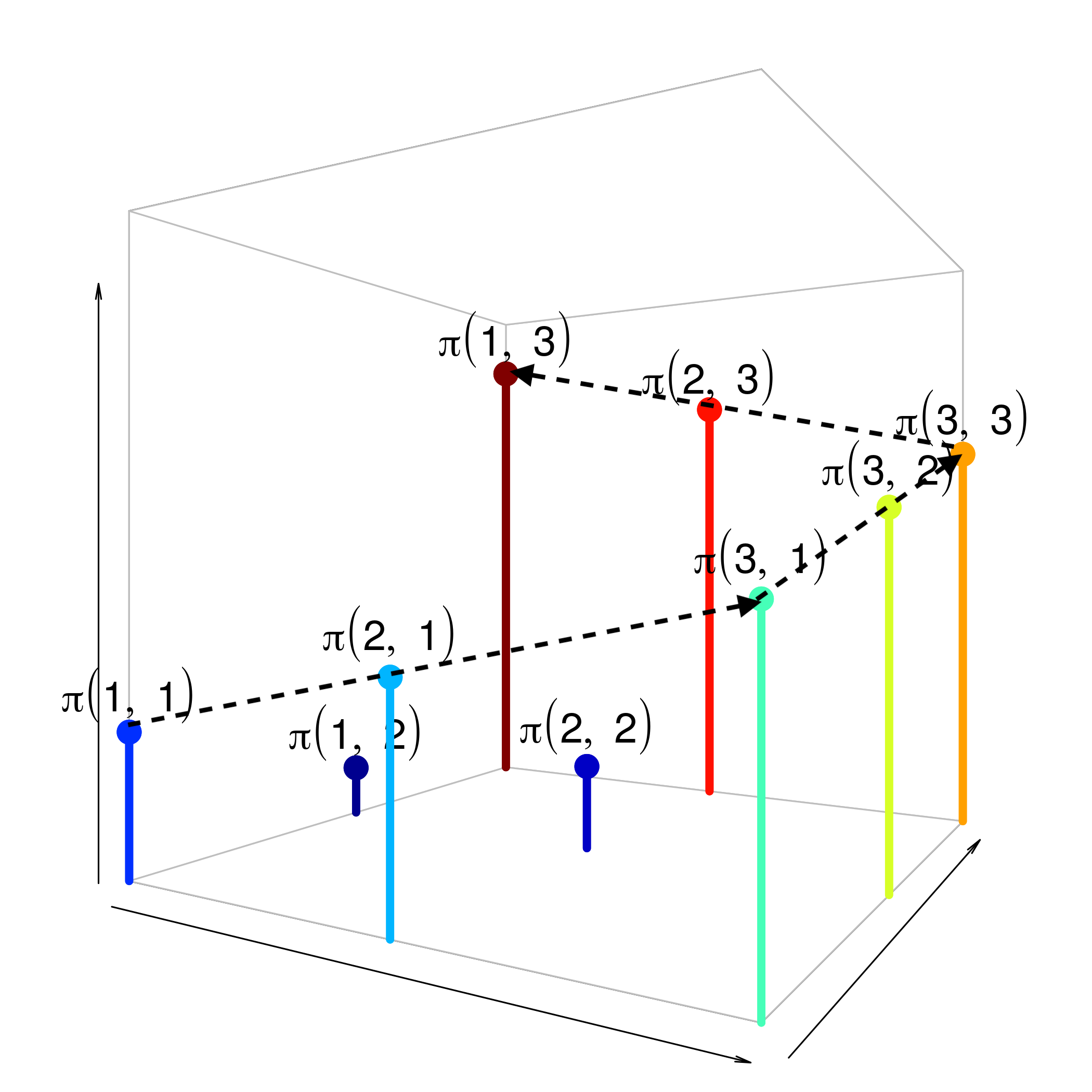}
    \includegraphics[width=0.4\linewidth]{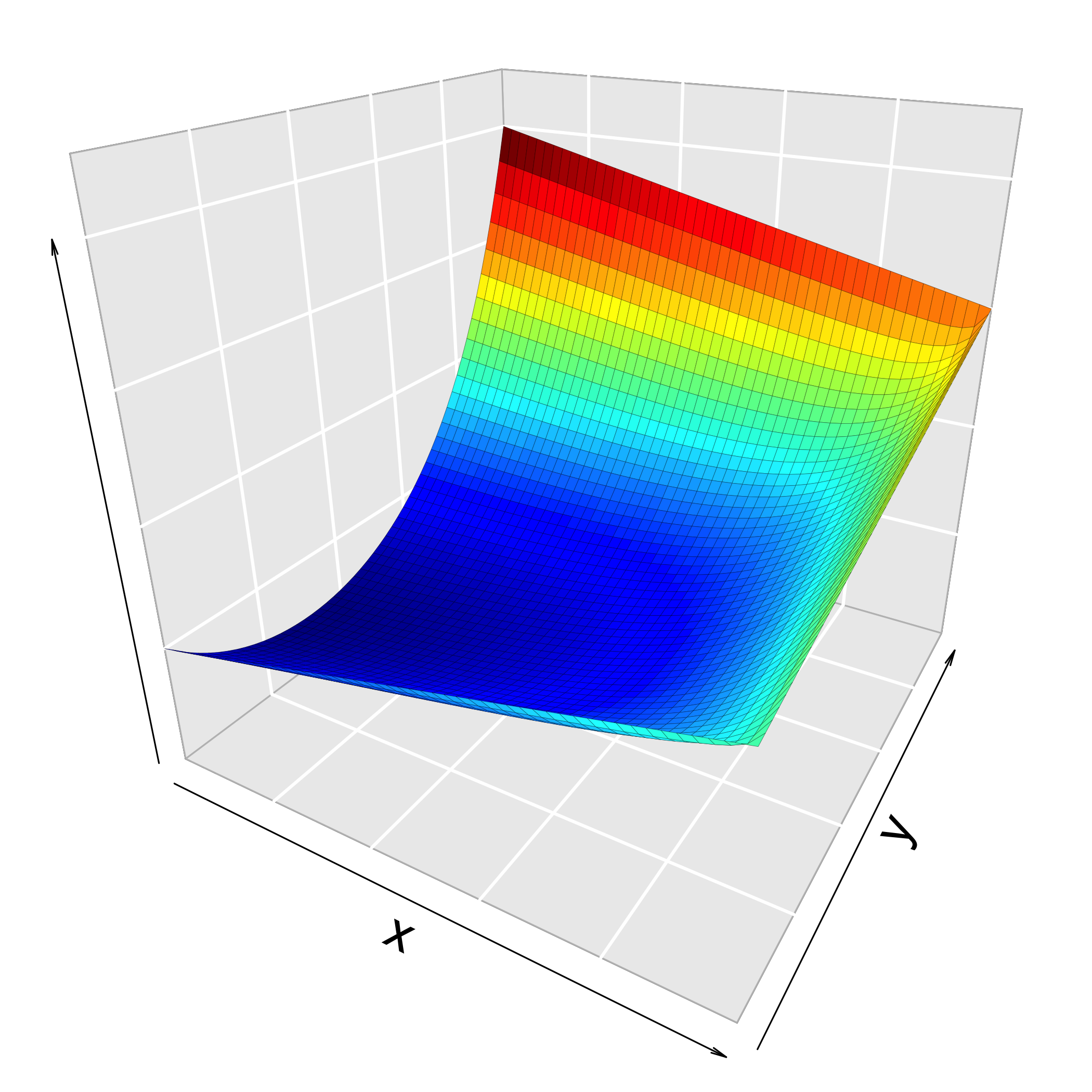}
    \caption{Left: a unimodal distribution defined on $\cX = \{1, 2, 3\}^2$ (the height of each bar is the value of $\log \pi$). Right: a continuous analogue of the distribution shown in the left panel.} 
    \label{fig:uni}
\end{figure}

\section{Simulation Settings and Results}\label{supp:sim}
\subsection{Weighted Permutations}\label{supp:sim.perm}
Let $\cX$ be the collection of all possible permutations  of $\{1, \dots, p\}$. 
For each permutation $\tau \in \cX$, let $\tau(i)$ be the index of the variable that has the $i$-th position in $\tau$  and $\tau^{-1}(k)$ be the ranking of the $k$-th variable. 
Assume that 
$$\pi(\tau) \propto \prod_{k=1}^p  W(k,  \tau^{-1}(k)),$$ 
where $W \colon \cX^2 \rightarrow \hbbR$ is a positive matrix, and 
$$\log W(k,  j) =  -  \eta \phi_k | j - \mu_k |  \log p, \quad \forall \, j, k \in \{1, \dots, p\}, $$  
for some $\eta > 0$, $\phi_k > 0$ and $\mu_k \in [1, p]$. We interpret $\eta$ as the signal-to-noise ratio (SNR) parameter.  
Let $\cN(\tau)$ be the set of all permutations that can be obtained from $\tau$ by a transposition; thus, $|\cN(\tau)| = p(p-1)/2$.  
We simulate $W$ in two ways such that $\pi$ has one unique mode at $\tau^* = (1, \dots, p)$. 
In Scenario I, we draw $\mu_k \sim \mathrm{Unif}(k-0.1, k+0.1)$  and $\phi_k \sim \mathrm{Unif}(0.5, 1)$ for $k = 1, \dots, p$, independently. 
In Scenario II, we draw $\mu_k \sim \mathrm{Unif}(k-0.5, k+0.5) $ and $\phi_k \sim \mathrm{Unif}(0.1, 1)$  for $k = 1, \dots, p$, independently.  
We explain in the next paragraph why $\pi$ is unimodal in both scenarios. 
Loosely speaking, in Scenario I, $\pi$ decays at roughly the same rate in every direction, while in Scenario II,  $\pi$ tends to be more irregular and decay very slowly in some directions, and some states in $\cN(\tau^*)$ can have posterior probabilities comparable to $\pi(\tau^*)$. Thus, Scenario II represents a ``weakly unimodal'' setting where Assumption~\ref{A1} is very likely to fail to hold for large $p$.  
We fix $p = 100$, and for each scenario, we generate two instances of $\pi$, one with $\eta = 1$  and the other with $\eta = 2$. 

In both scenarios,  observe that for each $k$,  $j \mapsto W(k, j)$ is unimodal with mode at $j = k$.   
To see that $\pi$ is guaranteed to be unimodal with mode at $\tau^*$, consider any permutation $\tau \neq \tau^*$. Let 
$$j_1 = \min\{ 1 \leq i \leq p \colon \tau^{-1}(i) \neq i \},  \text{ and }  j_{\ell} = \tau^{-1}(j_{\ell - 1}) \text{ for }  \ell = 2, 3, \dots.$$  
Observe that $j_2 > j_1$.  If $j_3 < j_2$, then we can swap the ranks of variables $j_1$ and $j_2$ and the new permutation will have a larger posterior probability.  If $j_3 > j_2$, let $m  = \min\{ \ell \geq 1  \colon j_{\ell + 1} < j_\ell   \}$. Observe that $m < p$ and  $j_{m+1} \geq j_1$. Hence, there exists some $1 \leq s \leq m - 1$ such that $j_{s } \leq j_{m + 1} < j_{s + 1}$. That is, $j_{s}  \leq \tau^{-1}(j_m) < \tau^{-1}(j_{s}) <  j_m$, which shows that we can increase the posterior probability by swapping variables $s$ and $m$. 

We compare RWMH and Algorithm~\ref{alg:iit} with the following five choices of $h$: $\htgs, \hmin, \hsq$, $h_{0.3}(u) = u^{0.3}$ and $h_{0.4}(u) = u^{0.4}$. 
For  $k = 1, \dots, p$, let $f_k(\tau) = \tau^{-1}(k)$ and consider the estimator $\hat{f}_k(t, \omega)$ defined in~\eqref{eq:sni}  (for RWMH, $\omega \equiv 1$, and for IIT samplers, $\omega = \pi / \pi_h$).  
The number of iterations $t$ is chosen to be $ 7\times 10^5$ for RWMH and $10^3$ for each IIT sampler so that the wall time used by each algorithm is about the same  
(the number of posterior evaluations does not dominate the empirical complexity in this study).
For each instance of $\pi$, we run each sampler $100$ times and then calculate the variance of $\hat{f}_k$ for each $k$. 
Boxplots for $\{ \Var(\hat{f}_k) \}_{k=1}^p$ are shown in Figure~\ref{fig1}. 
The simulation code is written in \texttt{R}.  

From Figure~\ref{fig1}, one can see that in Scenario I, the advantage of IIT schemes is overwhelming. In Scenario II with SNR $=1$, only the IIT scheme $\hsq$  is arguably slightly better than RWMH, and the IIT schemes $\htgs, \hmin$ and $h_{0.3}$ are clearly less efficient than RWMH.   But when SNR increases to $2$, IIT schemes $h_{0.3}, h_{0.4}, \hsq$ all perform significantly better than RWMH.  

\begin{figure*}[t!]
\centering
\includegraphics[width=0.48\linewidth]{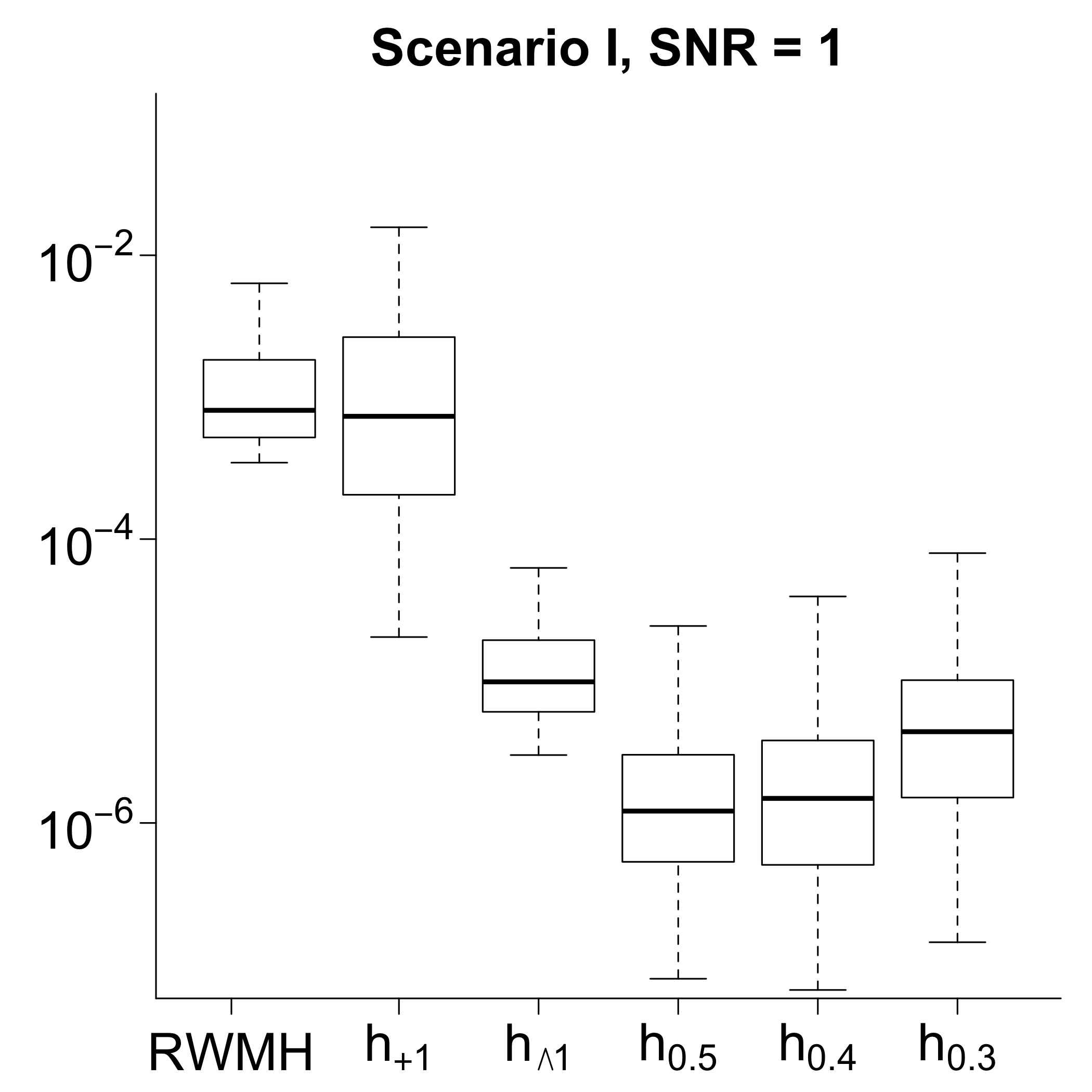} \hspace{0.2cm}
\includegraphics[width=0.48\linewidth]{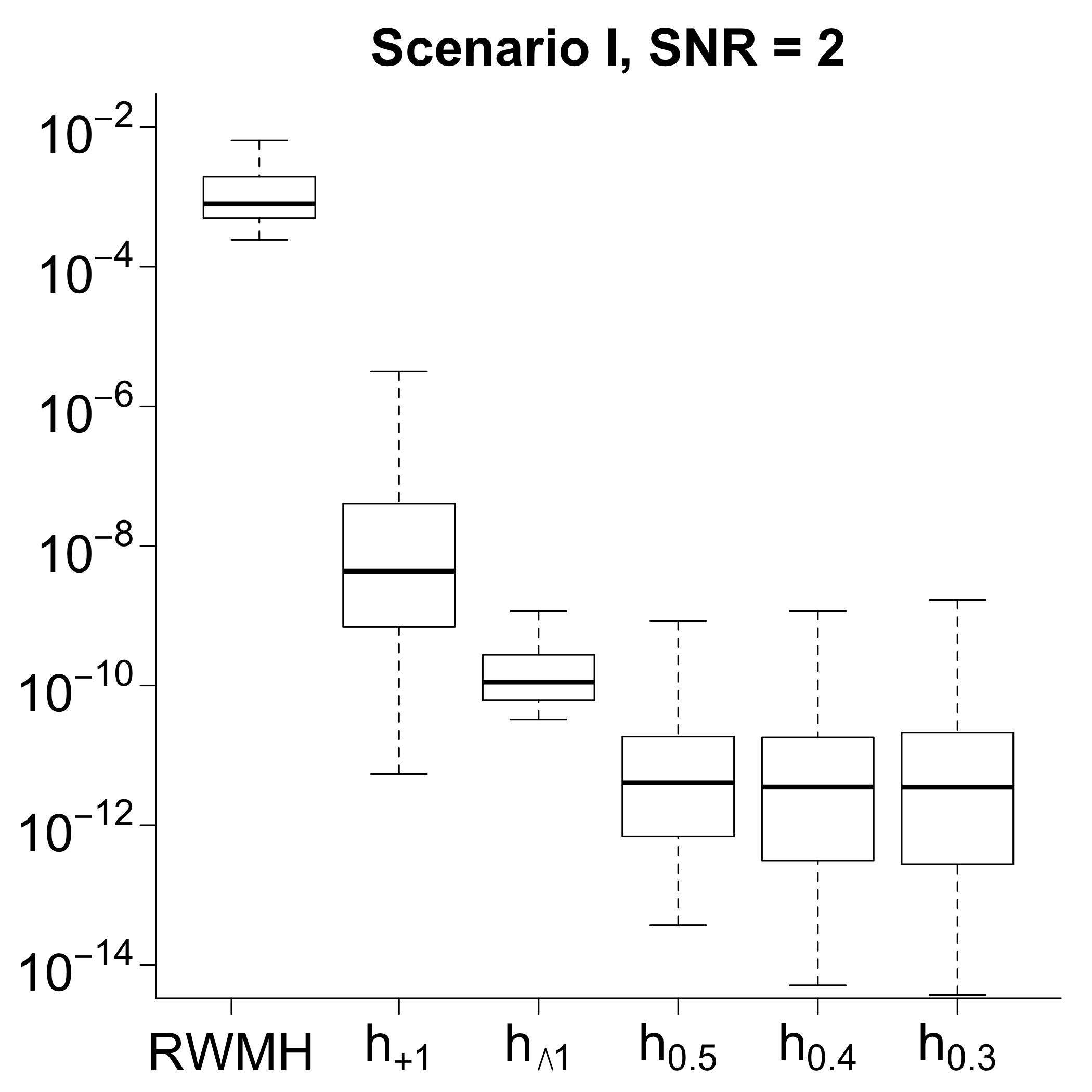}\\
\includegraphics[width=0.48\linewidth]{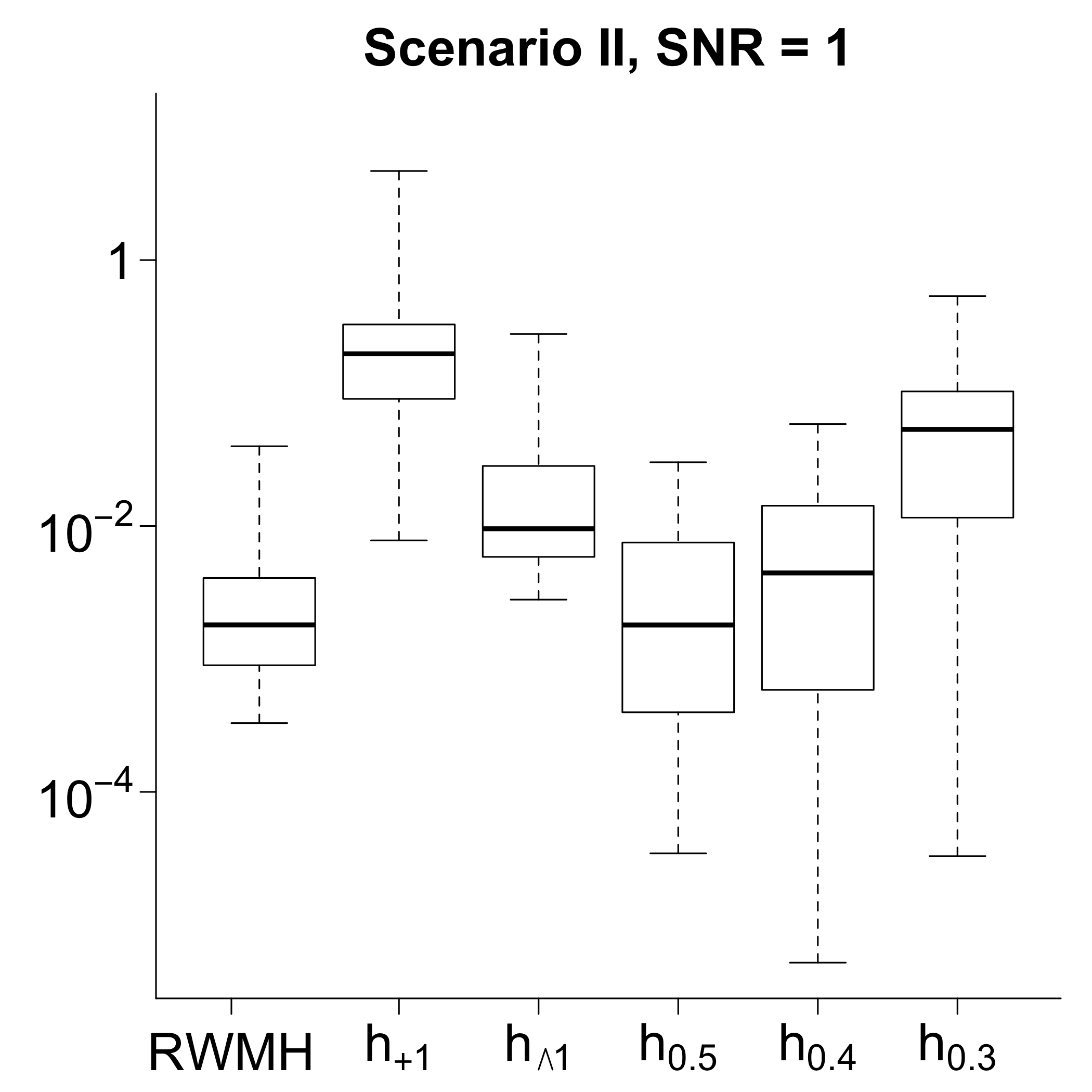}\hspace{0.2cm}
\includegraphics[width=0.48\linewidth]{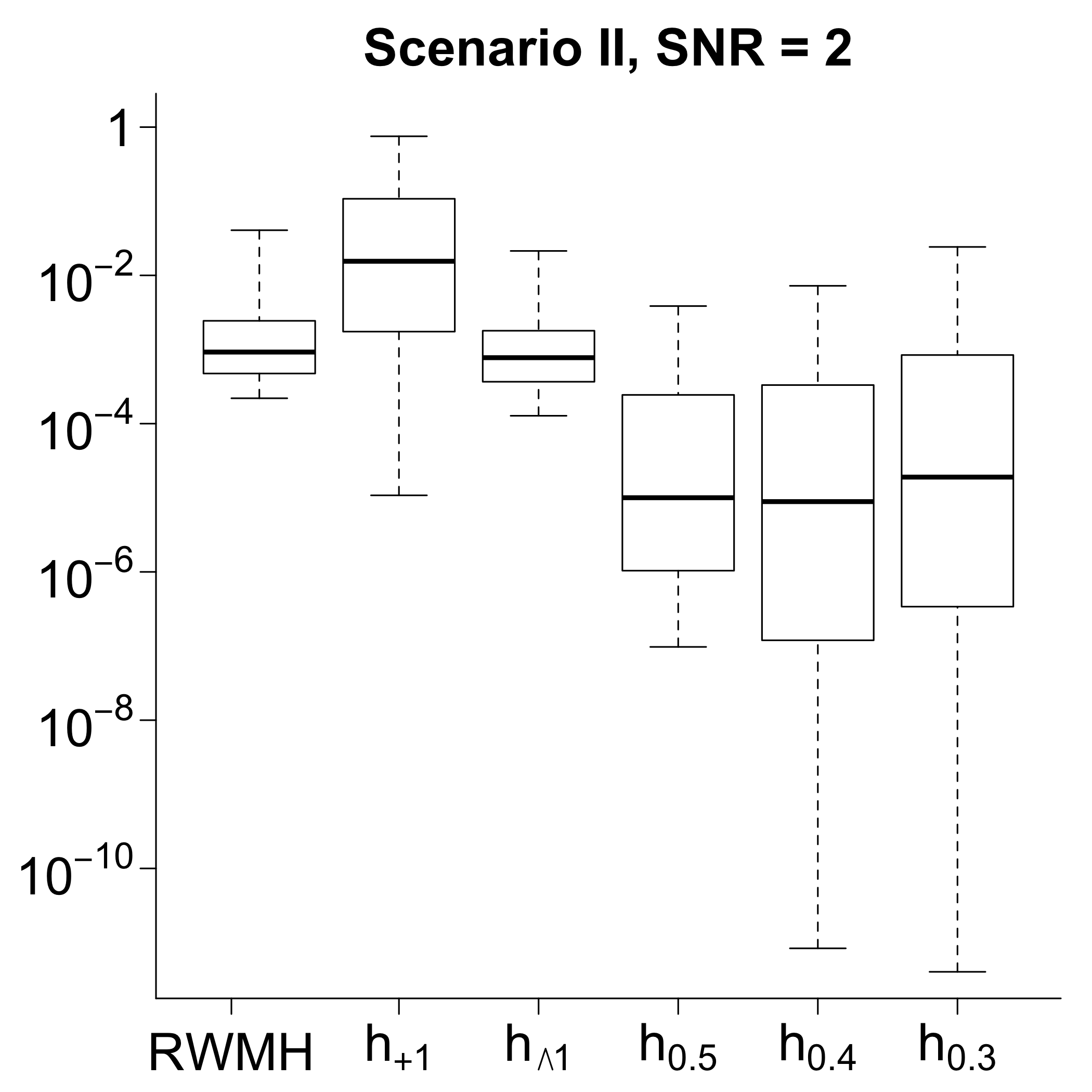} \\
\caption{ Each plot corresponds to one instance of $\pi$ for the ``weighted permutations'' problem with $p = 100$. 
Each box gives the empirical distribution of  $\{ \Var(\hat{f}_k) \}_{k=1}^p$ for one sampler, where $\Var( \hat{f}_k)$ is calculated from 100 independent runs of the algorithm. Each RWMH run has $ 7\times 10^5$ iterations, which takes about $10.3$ seconds; each IIT run has $10^3$ iterations, which takes about $10.1$ seconds.  
  }\label{fig1} 
\end{figure*}

\subsection{Variable Selection}\label{supp:sim.var}
We first describe how we simulate the data. Let $L \in \bbR^{n \times p}$ denote the design matrix. To mimic complex real-world problems with correlated predictors, we sample each row of $L$ independently from $N(0,  \Sigma )$ where $\Sigma_{ij} = e^{-|i - j|}$ for each $(i, j)$. Generate the response $Y \in \bbR^n$ by 
\begin{align*}
Y_i \overset{\mathrm{i.i.d.}}{\sim} N( \sum\nolimits_{j=1}^{j = s^*} \beta_j L_{i j},  1), 
\end{align*}
where $s^* = 20$  is the number of ``causal'' predictors. Generate $\beta$ by 
\begin{align*}
\beta_i  \overset{\mathrm{i.i.d.}}{\sim} \mathrm{SNR} \sqrt{\frac{\log p}{n}} \mathrm{Unif}( (2, 3) \cup (-3, -2) ). 
\end{align*}
We use $ \mathrm{SNR} = 1, 2, 3$, and for each value of SNR, we simulate $100$ replicates of $(L, Y)$.  This simulation setting is very similar to that used in~\citet{yang2016computational} (the main difference is that they used $s^* = 10$).   

As in Example~\ref{ex:var1},  denote the state space by $\cX = \{0, 1\}^p$ (a hard threshold on the model size is unnecessary in our simulation since an MCMC sampler never visits extremely large models in all of our runs). 
We follow~\citep{yang2016computational} to calculate the posterior distribution by   
\begin{equation}\label{eq:var.pi}
\pi(x) \propto p^{- c_0 \norm{x}_1 }  \frac{ (1 + g)^{- \norm{x}_1 / 2} }{ (1 + g(1-R_x^2))^{n/2} }, 
\end{equation}
where $R^2_x$ is the usual R-squared statistic for the model $x$, and $c_0$ is the hyperparameter in the sparsity prior on $x$,  and $g$ is the hyperparameter in the g-prior.  We fix $1 + g = p^{3}$ and $c_0 = 2$ in our simulation. 

We consider four IIT samplers with weighting scheme $\htgs, \hmin, \hsq$ and $h_{0.3}$. For all of them,  we use the neighborhood relation $\cN^1$, which we recall only includes ``addition'' and ``deletion'' moves. Note $|\cN^1(x)|  = p$ for all $x$. 
For comparison, we consider the standard add-delete-swap MH sampler (denoted by ADS) with proposal distribution 
\begin{align*}
K_{\rm{ads}}(x, y) = 0.4  \frac{ \ind_{\cN_{\rm{add}}(x)} (y) }{|\cN_{\rm{add}}(x)  | } + 0.4  \frac{ \ind_{\cN_{\rm{del}}(x)} (y) }{|\cN_{\rm{del}}(x)  | }  + 0.2  \frac{ \ind_{\cN_{\rm{swap}}(x)} (y) }{|\cN_{\rm{swap}}(x)  | }.  
\end{align*}
This is slightly different from a RWMH that proposes uniformly from $\cN^2(x)$, since we fix the proposal probabilities of three types of moves to $0.4, 0.4, 0.2$. In the main text, we simply refer to this ADS sampler as RWMH. 
  
For each simulated data set, we initialize all samplers at the same  model consisting of 10 random covariates, and run ADS for $5 \times 10^6$ iterations and each IIT sampler for $5 \times 10^3$ iterations. 
Each IIT  run takes less than $20$ seconds while ADS takes about $2,000$ seconds (it is likely that the code for ADS sampler can be further optimized, but due to overhead cost, our setting should already be very fair to ADS). The code for IIT sampling is written in \texttt{C++}. 

Before we discuss the performance of five samplers, we investigate the shape of $\pi$ to verify whether the assumptions used in our theory approximately hold for this variable selection problem. 
Calculating $\pi$ at each $x \in \cX$ is impossible since $|\cX| = 2^p \approx 10^{1505}$. 
So we empirically check the multimodality of $\pi$ by combining sample paths of all IIT runs. 
First, we count the number of local modes of $\pi$. We say a point $x$ is a local mode if $\pi(x) > \argmax_{y \in \cN^1(x)} \pi(y)$.  When SNR $=2$ (intermediate SNR),  among the 100 replicates, the number of local modes of $\pi$ (that have been visited by any sampler) ranges from $1$ to $12$  with an average of $5.3$.  
When SNR $=1$ or SNR $=3$, the number of local modes tends to be much smaller; see Figure~\ref{fig:local.modes}. 
Next, we consider the  quantity 
\begin{equation}\label{eq:def.nu}
\nu(x) = \max_{y \in \cN^1(x)} \log_p\left(  \frac{  \pi(y) }{ \pi(x) } \right). 
\end{equation}
Assumption~\ref{A1} would hold  if  $\nu(x) > 1$ for each $x \neq x^*$, and Theorem~\ref{th:two} suggests that we need at least $\nu(x) > 2$ at most $x$ so that IIT schemes can achieve optimal convergence rates. We find that $\nu(x)$ is indeed  large and greater than $2$ at most points in all the three SNR cases. See Figure~\ref{fig:nu} and Table~\ref{table:nu}. 

For the variable selection problem, the accuracy of posterior estimation largely depends on whether the sampler is able to find the ``best'' model in its search. So, we compare the performance of samplers using $\Tmap$, where  $\xmap$ denotes the model with the largest posterior probability that has been visited by any sampler (ADS or IIT) and $\Tmap$ denotes the number of iterations taken by a sampler to find $\xmap$. 
See Figure~\ref{fig:violin} for the distribution of $\Tmap$ for each sampler.

When SNR $=3$, $\pi$ tends to be unimodal with true model $x^*$ being the unique mode. In this case, IIT schemes $\hmin, \hsq$ and $h_{0.3}$ find $x^*$ very quickly and within $\approx 30$ iterations in most cases. This is expected since we initialize the sampler at some random $x$ with $\norm{x}_1 = 10$ (so it takes about $10$ iterations to remove the non-influential covariates and $20$ iterations to add the influential ones). 
When SNR $=1$, $\pi$ is still likely to be unimodal, but the mode may be the empty model (or some other model with very small size) due to the weak signal size. This is why ADS sometimes can be very efficient since it only needs to remove the 10 non-influential covariates in the initial model. 
However, we still observe that $\hsq$ and $h_{0.3}$ have better performance. 
The intermediate SNR case is the most challenging since $\pi$ is usually multimodal. We observe that the advantage of IIT schemes $\hsq$ and $h_{0.3}$ is very significant in the sense that ADS often fails to find $\xmap$ in $5 \times 10^6$ iterations; see Table~\ref{table:hit}. 
The best sampler in this case appears to be $h_{0.3}$ since it is most robust. 

The IIT scheme $\htgs$ performs worse than ADS in all the three cases. \citet{zanella2019scalable} also noted that the naive TGS sampler (i.e., IIT scheme $\htgs$) may converge slowly for variable selection, and thus they proposed to use a weighted TGS algorithm, which yielded excellent performance. In addition, they used Rao-Blackwellization to reduce the variance of posterior estimation. We note that both schemes are highly effective but rely on the specific structure of $\cX$ in variable selection. The same Rao-Blackwellization scheme can also be applied to IIT, but we did not observe significant improvement. 

\begin{figure}[t!]
\centering
\includegraphics[width=0.32\linewidth]{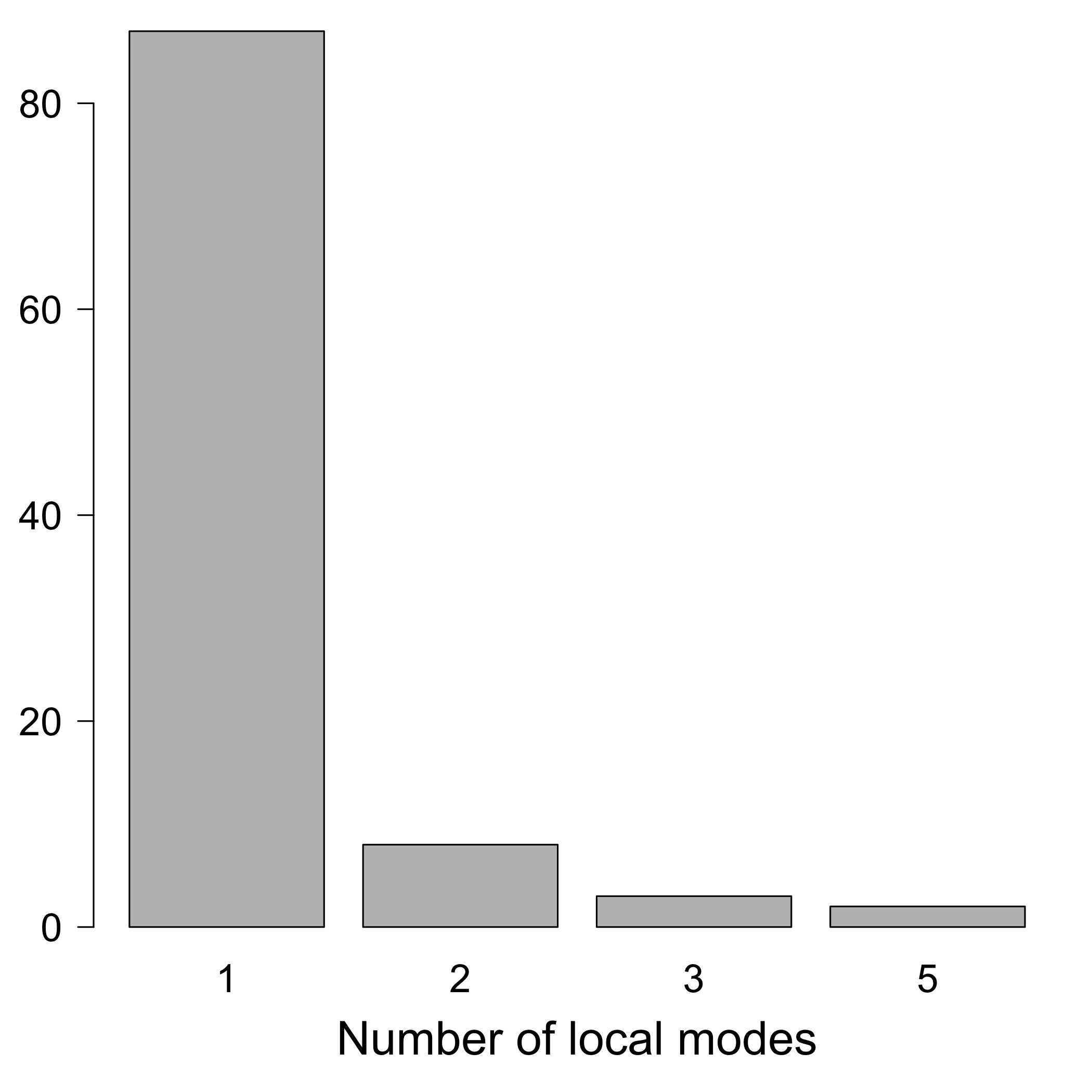}
\includegraphics[width=0.32\linewidth]{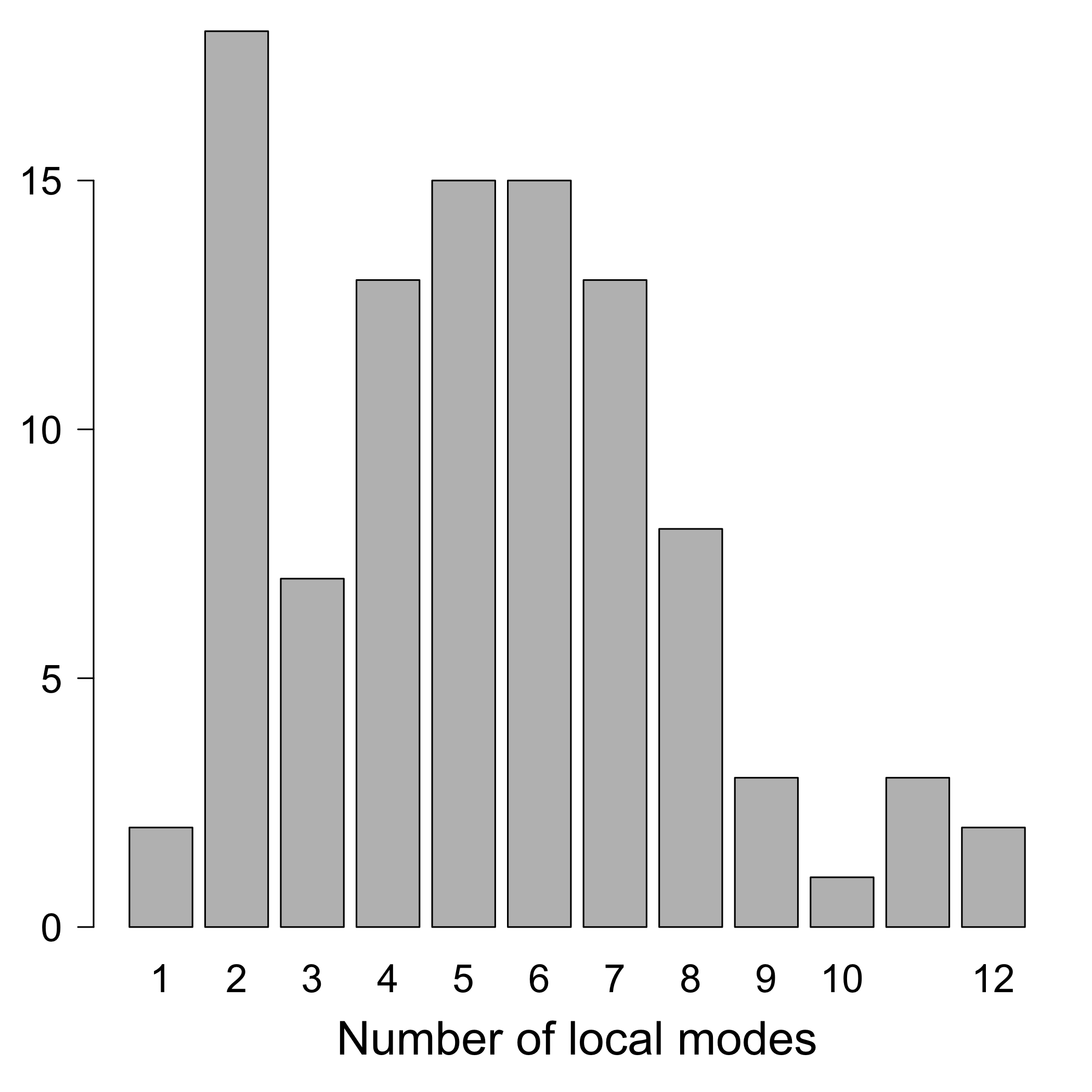}
\includegraphics[width=0.32\linewidth]{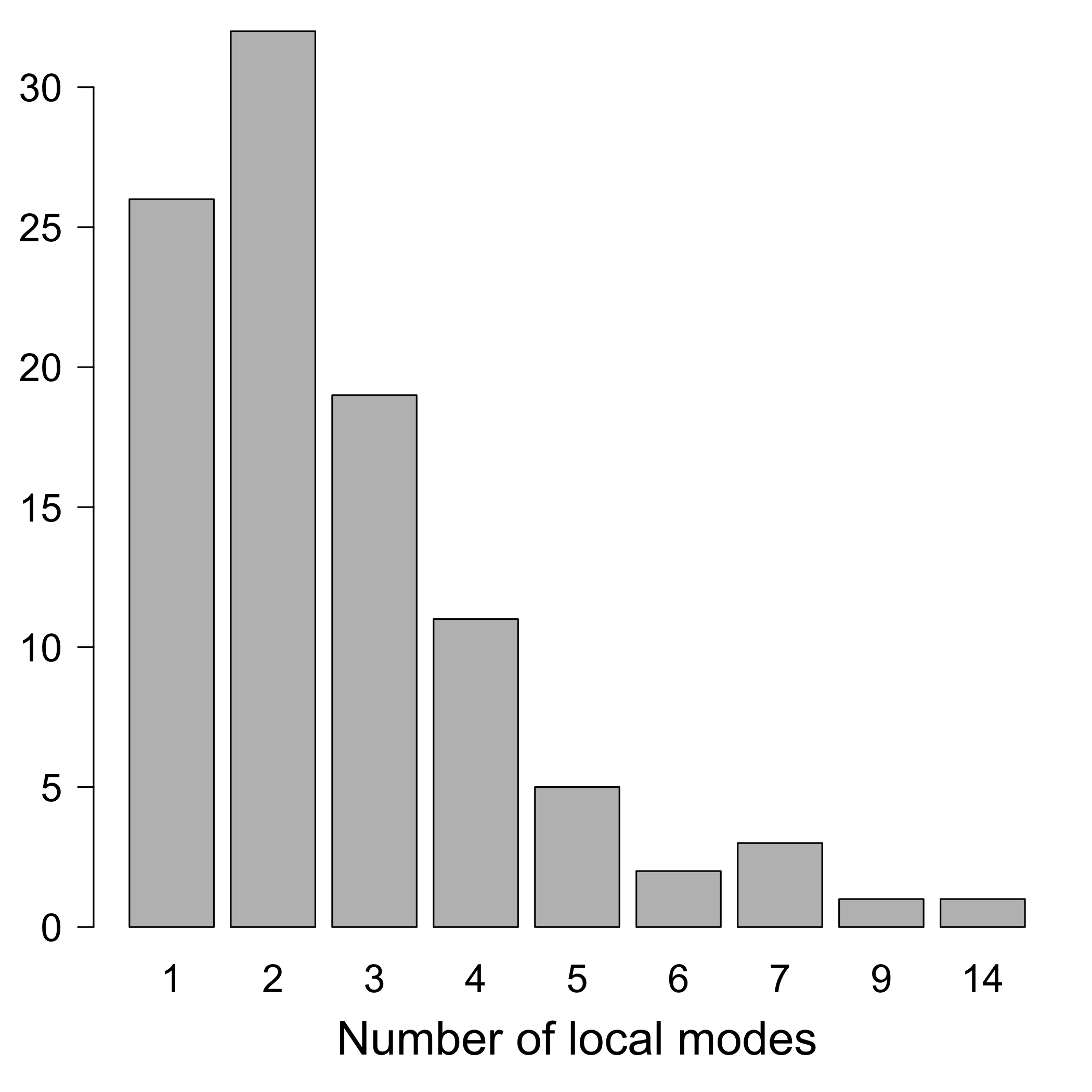}
\caption{Distribution of the number of local modes in the variable selection simulation study. 
Left: SNR $=3$; middle: SNR $=2$; right: SNR $=1$. For each value of SNR, we count the number of local modes of $\pi$ (visited by any IIT sampler) for 100 simulated data sets.} \label{fig:local.modes}
\end{figure}

\begin{figure}[t!]
\centering
\includegraphics[width=0.32\linewidth]{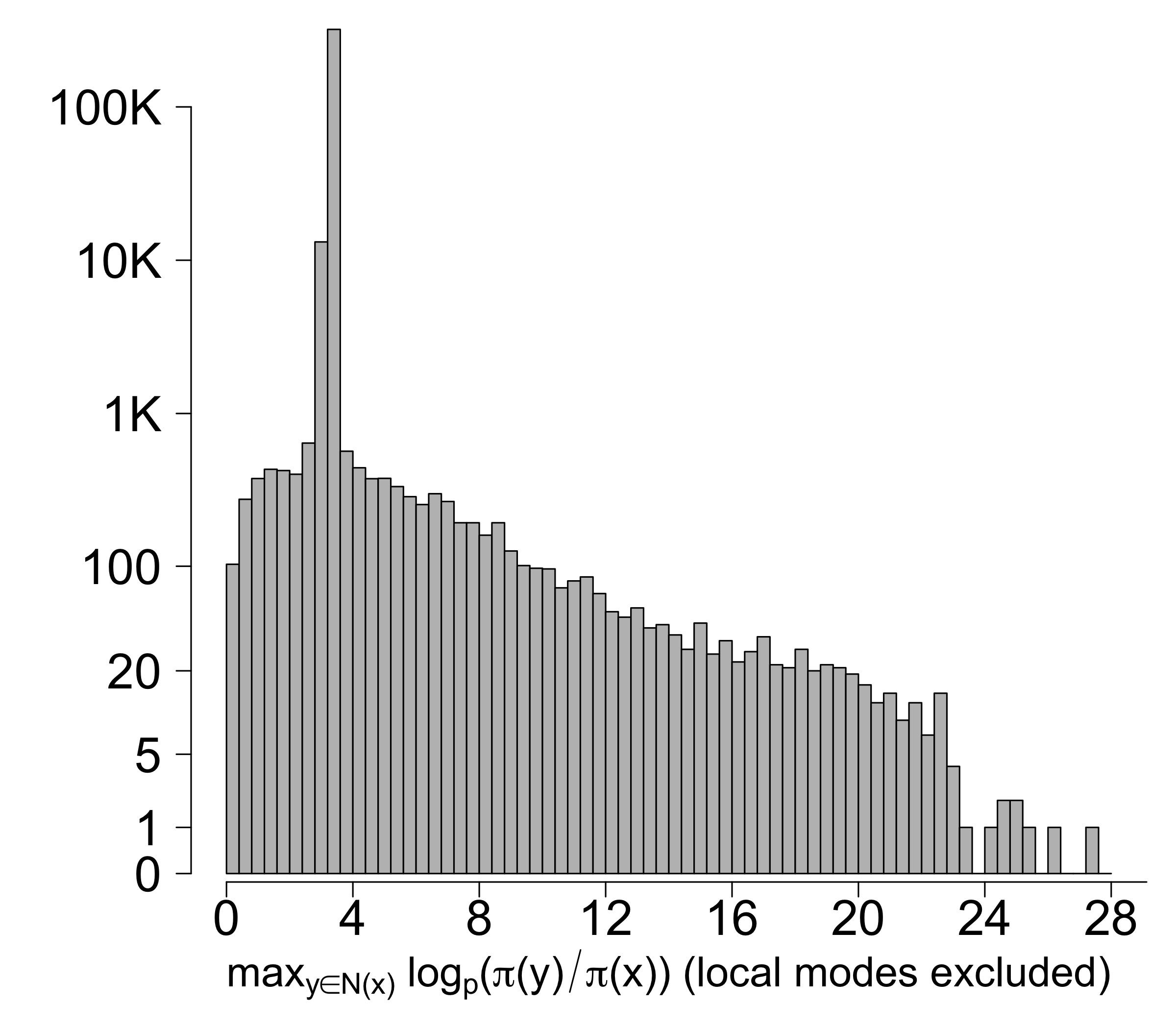}
\includegraphics[width=0.32\linewidth]{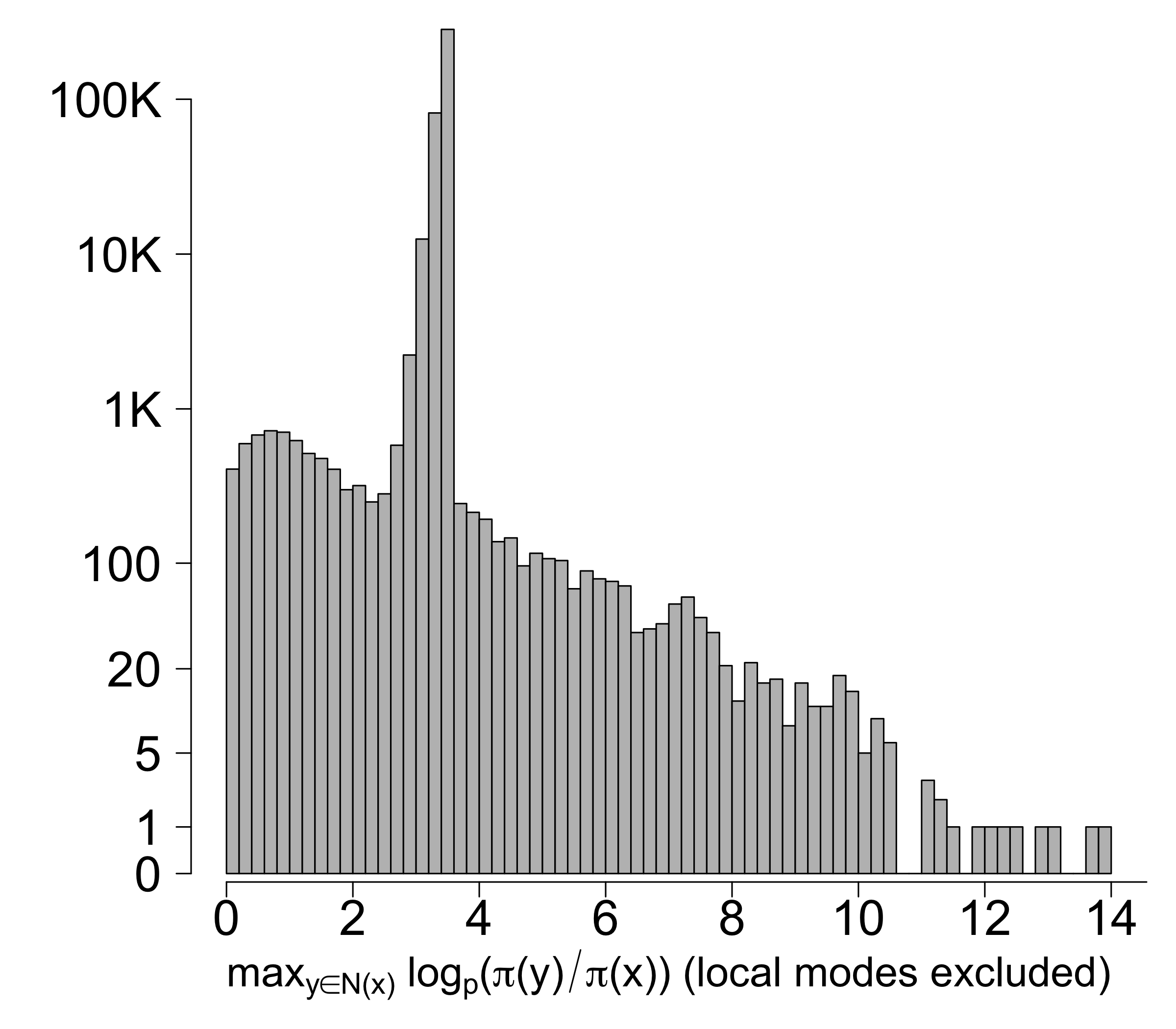}
\includegraphics[width=0.32\linewidth]{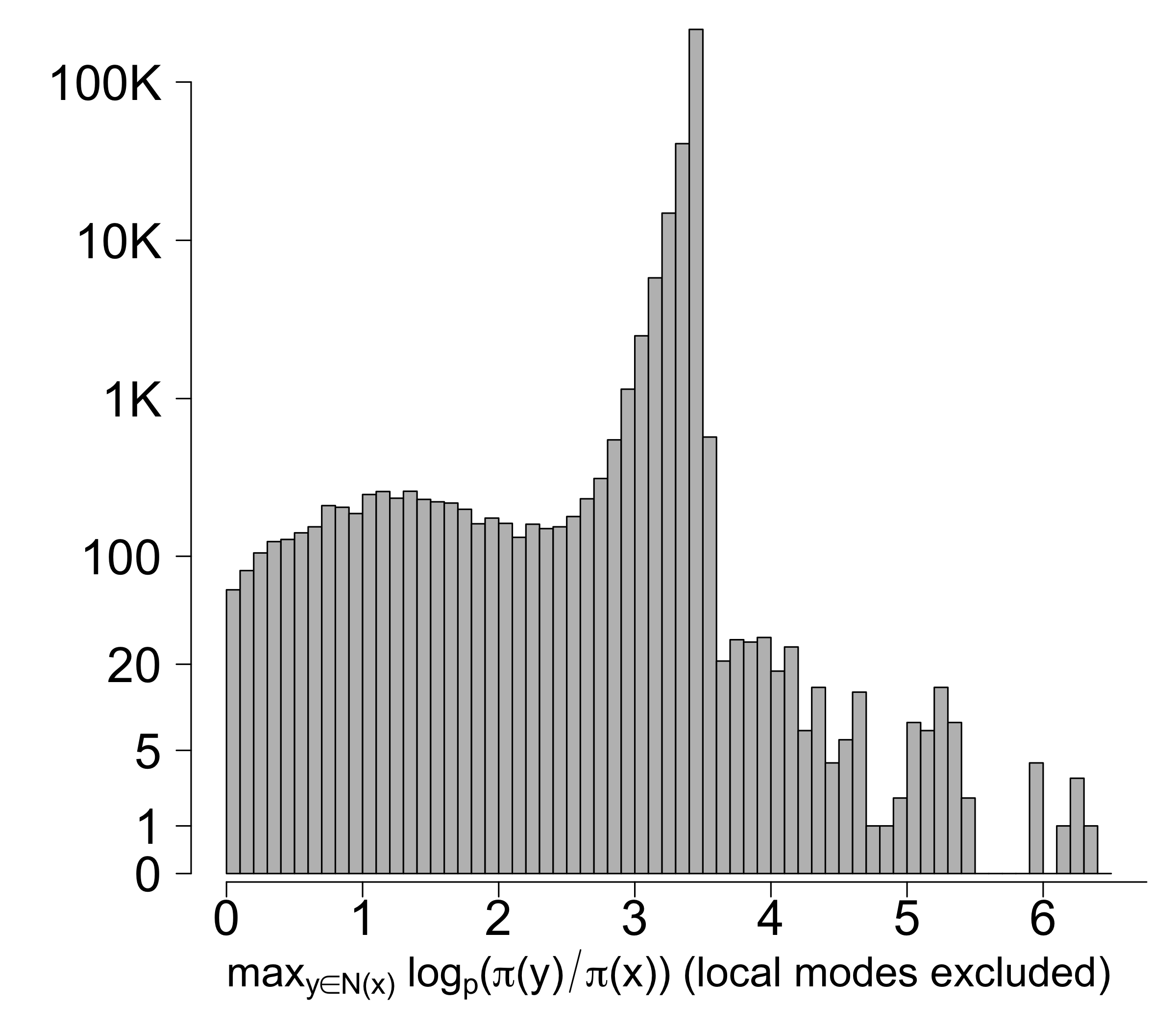}
\caption{Distribution of $\nu(x)$ 
in the variable selection simulation study.  Left: SNR $=3$; middle: SNR $=2$; right: SNR $=1$. For each value of SNR, we pool together the sample paths of IIT samplers for all 100 simulated data sets.} \label{fig:nu}
\end{figure}

\begin{table}[t!]
\centering
\begin{tabular}{cccc}
\hline
SNR   & $\nu(x) > 1$ & $\nu(x) > 2$ & $\nu(x) > 3$     \\
\hline
SNR $=3$ & 99.8\% & 99.5\% & 98.7\% \\ 
\hline
SNR $=2$ & 99.1\% & 98.5\% & 97.5\% \\
\hline
SNR $=1$ & 99.4\% & 98.7\% & 97.5\% \\
\hline 
\end{tabular}
\caption{Distribution of $\nu(x)$. 
For each SNR, we report the percentage of points $x$ with $\nu(x) > c$ for $c = 1, 2, 3$ using the data shown in Figure~\ref{fig:nu} (local modes are included in this calculation). }\label{table:nu}
\end{table}

\begin{table}[h!]
\centering
\begin{tabular}{cccccc}
\hline
SNR   & ADS &  $\htgs$  &  $\hmin$ &  $\hsq$ & $h_{0.3}$   \\
\hline
SNR $=3$ &0 & 2 & 0 &  0  & 0 \\ 
\hline
SNR $=2$ & 38 & 72 & 15 & 10 &  1 \\
\hline
SNR $=1$ & 1 & 10 & 1 & 0 & 0 \\
\hline 
\end{tabular}
\caption{Number of experiments (out of 100) where the sampler fails to find $\xmap$ in the entire run. This supplements Figure~\ref{fig:violin}. }\label{table:hit}
\end{table}

\begin{figure}[t!]
\centering
 {\Large SNR $=3$} \\
 \includegraphics[width=0.7\linewidth]{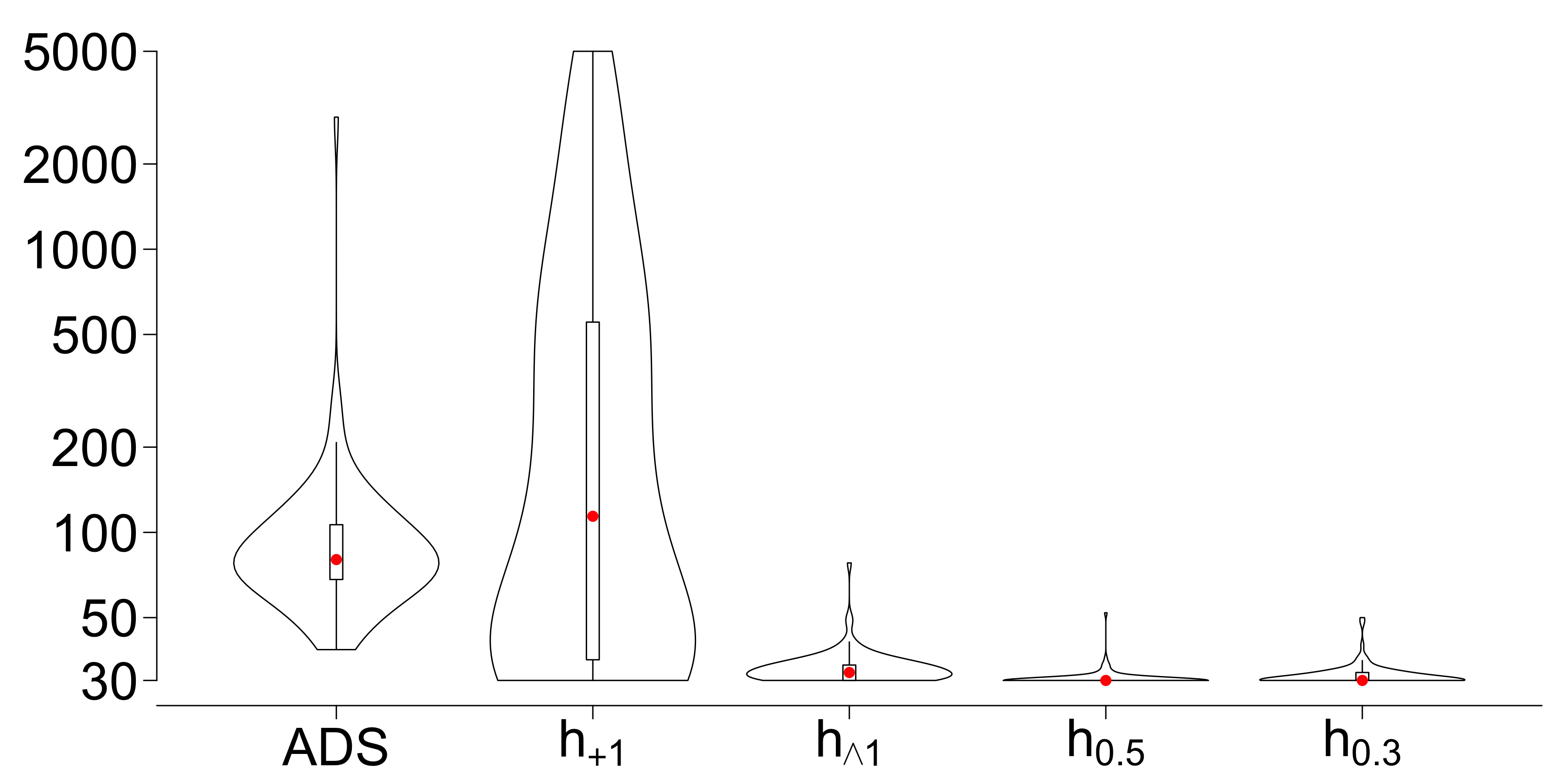}\\ 
\bigskip
 {\Large SNR $=2$}   \\
 \includegraphics[width=0.7\linewidth]{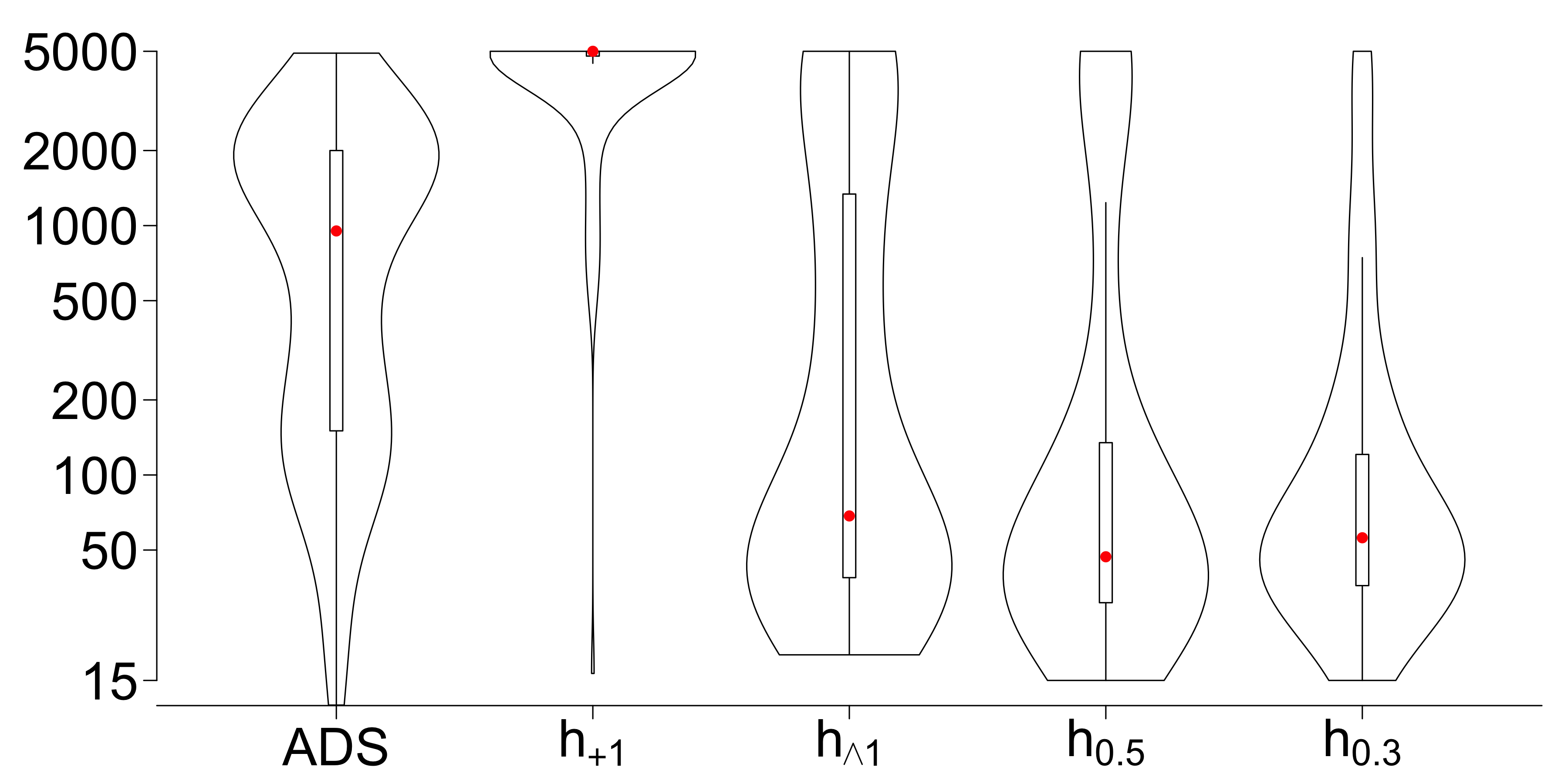}\\
\bigskip
  {\Large SNR $=1$} \\
 \includegraphics[width=0.7\linewidth]{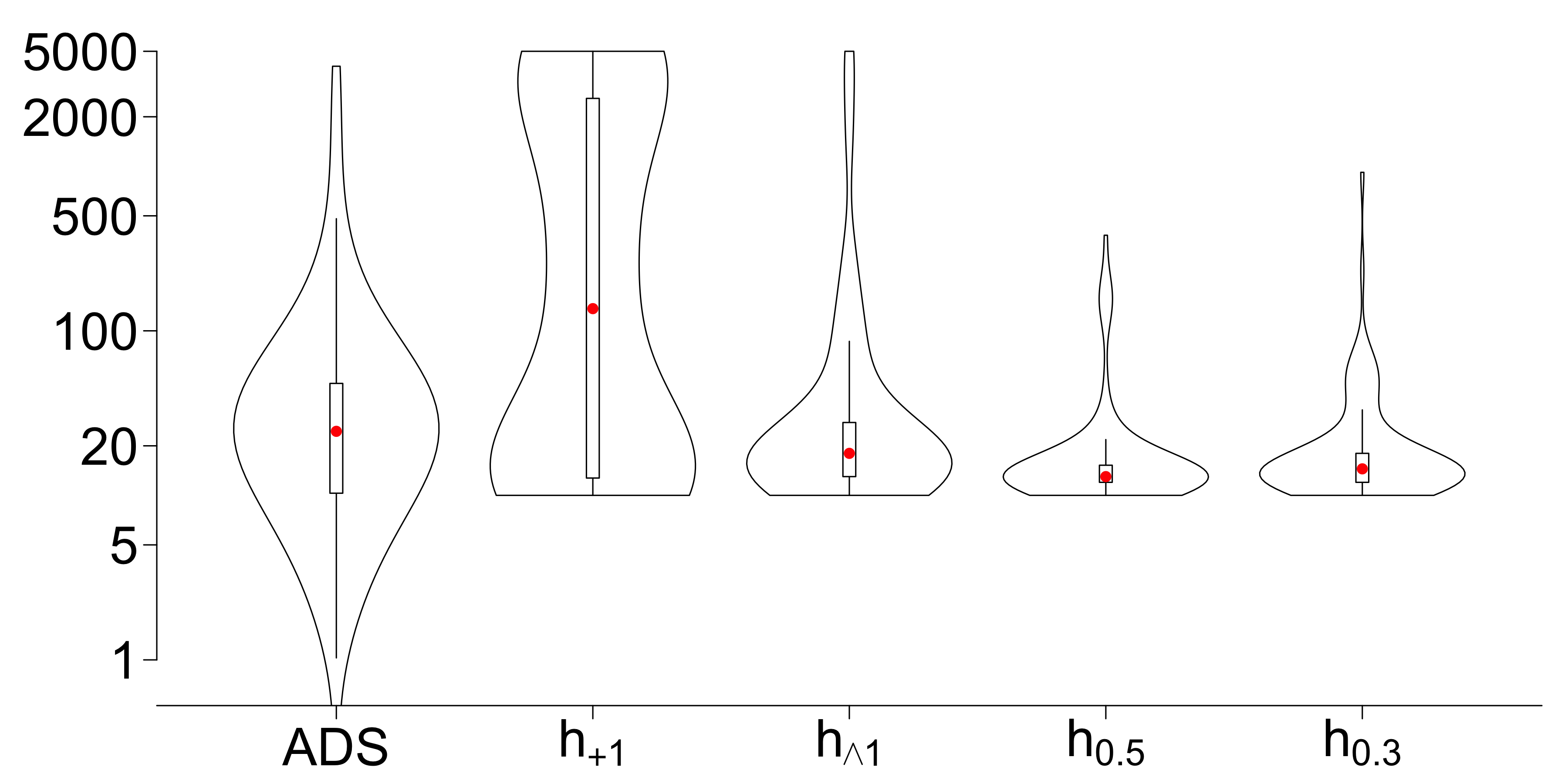} \\
\caption{Violin plots for the distribution of  $\Tmap$ in 100 replicates. The numbers on the y-axis correspond to $\Tmap$ for IIT samplers and $\Tmap  \times 10^{-3}$ for ADS. 
 $\Tmap$ is truncated at $5 \times 10^3$ for IIT  and $5 \times 10^6$ for ADS. 
The plot is made using \texttt{R} package \texttt{vioplot}. 
A boxplot is drawn inside each violin: the red dot denotes the median and the rectangle corresponds to the interquartile range. Note that the width of each violin is forced to be the same. } \label{fig:violin}
\end{figure}
\clearpage
\newpage 

\section{Examples in the Main Text} \label{supp:exs}
\subsection{Details of Example~\ref{ex:bad}}\label{supp:bad}
We show how to derive the expression for $\omega$. First, observe that 
\begin{align*}
    1 = \sum_{k=1}^p \pi(k) \leq \pi(0) + \pi(1) + \pi(1) r^{-1}  + \cdots + \pi(1) r^{-(p-1)}  \leq  \pi(0) + \frac{\pi(1)}{1 - r^{-1}}.   
\end{align*}
Since $\pi(0) = \pi(1)$ and $r \rightarrow \infty$, we get $\pi(0) = \pi(1) \sim 1/2$, and for $k \geq 2$, $\pi(k) \lesssim r^{-(k-1)}/2$.  
For $h(u) = u^a$ with some fixed $a > 0$, we have 
\begin{align*}
   \sum_{x \in \cX} \pi(x)^{2a} Z_h(x)  = \sum_{x \in \cX} \sum_{y \in \cN(x)} \pi(x)^a \pi(y)^a 
 = \sum_{k=0}^{p - 1}  2 \pi(k)^a \pi(k+1)^a. 
 \end{align*} 
Observe that as $k$ grows, $\pi(k)^a \pi(k+1)^a$ decreases at a rate not slower than $r^{-a}$. Therefore, as $r \rightarrow \infty$, we obtain that 
\begin{align*}
    \sum_{x \in \cX} \pi(x)^{2a} Z_h(x)  \sim 2 \pi(0)^a \pi(1)^a \sim 2^{1 - 2a}.  
\end{align*}
By Lemma~\ref{lm:pih}, 
\begin{align*}
    \pi_h(k) = \frac{ \pi(k)^{2a} Z_h(k) }{ \sum_{x \in \cX} \pi(x)^{2a} Z_h(x) }
   \sim \left\{ 
    \begin{array}{cc}
      \frac{ \pi(k - 1)^{a} \pi(k)^a } {2^{1-2a}},    &   k = 2, \dots, p,  \\
    1/2, &  k = 0, 1. 
    \end{array}
    \right.
\end{align*}
The importance weight $\omega$ then can be calculated by $\omega(k) = \pi(k) / \pi_h(k)$. 

\subsection{Details of Example~\ref{ex:counter}}\label{supp:counter}
Define a function $f \colon \cX \rightarrow \{-1, 1\}$ by letting $f(x) = 1$ if $x_1 = 1$ and $f(x) = -1$ if $x_1 = 0$. 
We show that $f$ is an eigenvector of $Q_h$ with eigenvalue $-2/p$, which proves $\gap(Q_h) \leq  2/ p$.  
We  use the notation $y^{(i)}(x)$  defined in~\eqref{eq:yi}. 
For any $x$ with $x_1 = 1$,  $f(y^{(i)}(x)) = 1$ for any $i \geq 2$, and $f(y^{(1)}(x)) = -1$. Hence, 
\begin{align*}
\sum_{x \in \cX} Q_h(x, y) f(y) = \sum_{i=2}^p Q_h(x, y^{(i)}(x) ) -  Q_h(x, y^{(1)}(x) )    +  Q_h(x, x) 
=  - 2 Q_h(x, y^{(1)}(x) ). 
\end{align*} 
Using $\pi(y^{(1)}(x)) = \pi(x)$ and~\eqref{eq:Q1},  we find that $Q_h(x, y^{(1)} (x)) = 2 / \bbE_\pi[Z_h]$.  
Since $|\cN^1(x)| = p$ for every $x$, we get $\bbE_\pi[Z_h] = 2 p$ from~\eqref{eq:EZ}, and thus $\sum_{x \in \cX} Q_h(x, y) f(y) = -2/p$. 
The case $x_1 = 0$ follows by an analogous calculation. 

\section{A Generic Mixing Time Bound for Decomposable Markov Chains}\label{supp:mix}

\begin{lemma} \label{LemmaDecompGeneric} 
Consider the transition kernel $P$ of a $\frac{1}{2}$-lazy  discrete-time Markov chain with unique stationary measure $\pi$ on state space $\mathcal{X}$. Denote by $\cX = \sqcup_{i=1}^{M} \cX_{i}$ a decomposition of the state space, and denote by $x_{i}^{*} \in \cX_{i}$ a privileged point satisfying $\pi(x_{i}^{*}) / \pi(\cX_i) > 0.65$ for each $i \in \{1,2,\ldots,M\}$. Define $\cX^{*} = \cup_{i=1}^{M} x_{i}^{*}$ to be the collection of all privileged points.  
Let $P|_{\cX_{i}}$ denote the trace of $P$ on $\cX_i$. 
Then there exists a universal constant $C$ so that 
\begin{align*}
 \tmix(P) \leq  C \,  \max \left\{    \gap_{\rm{min}}^{-1} \, M \, \log \frac{8 M}{  \pmin} ,  \,  \tmix(P|_{\cX^*}) \right\},
\end{align*} 
where $\gap_{\rm{min}} = \min_{i \in \{1,2,\ldots,M\} } \gap( P|_{\cX_{i}} )$ and $\pmin = \min_{x \in \cX} \pi(x)$ . 
\end{lemma}
\begin{proof}

Let $\bbP_x$ denote the probability measure corresponding to a Markov chain $\{X_t\}_{t \geq 0}$ with transition matrix $P$ and $X_0 = x$.  
Let $\{X^{(k)}_t\}_{t \geq 0}$ denote the trace of $\{X_t\}$ on $\cX_k$, as defined in Definition~\ref{def:trace}.   Let $\pi_i (\cdot) = \pi( \cdot ) / \pi(\cX_i)$ denote the stationary distribution of $P|_{\cX_k}$. 
By assumption, we have  $\bbE_{\pi_i}[ \ind_{ \{ x_i^* \} } ( X^{(i)}_t ) ] = \pi_i(x_i^*) \geq 0.65$.  
By the Chernoff-like bound \citet[Theorem 1.1]{LezaudChernoff}, 
\begin{equation} \label{EqOccBd}
    \bbP_x [|\{t \in \{1, 2, \dots, T\} \colon   X_{t}^{(i)} \neq x^*_i \}| \geq 0.5 T] \leq \frac{2}{\pi(x)} \exp \left ( - \frac{  \gap(P|_{\cX_i})}{56} T \right\}. 
\end{equation}

For $i \in \{1,2,\ldots,M\}$ and $T \in \mathbb{N}$, denote by $\mathcal{O}_{T}(i) = |\{0 \leq t \leq T \, : \, X_{t} \in \cX_{i}\}|$ the number of steps that the chain spends in $\cX_{i}$.  Denote by 
$$\mathcal{G}_{T} = \left\{i \in \{1,2,\ldots,M\} \, : \, \mathcal{O}_{T}(i) > 56 \, \gap_{\rm{min}}^{-1} \, \log\left( \frac{8 M}{  \pmin} \right) \right\}$$ 
the collection of indices for which this occupation time is large relative to the relaxation time.
By the union bound,    for all $x \in \mathcal{X}$:
\begin{equation} \label{EqOccBdPre} 
    \bbP_x \left [\forall \, i \in \mathcal{G}_{T}, \, |\{ 0 \leq t \leq T \, : \, X_{t}^{(i)} = x_{i}^* \}| > \frac{1}{2} \mathcal{O}_{T}(i) \right]  \geq 0.75. 
\end{equation}
By the pigeonhole principle, for $T \geq T_1 = (5 M)  56 \, \gap_{\rm{min}}^{-1} \, \log\left( \frac{8 M}{  \pmin} \right)$, we must have 
$$|\{0 \leq t \leq T \, : \, X_{t} \in \cup_{ i \in \mathcal{G}_T}  \cX_i  \}| \geq \frac{4}{5}T. $$ 
Combining this with inequality \eqref{EqOccBdPre}, we conclude that for all $T \geq T_1$, we have 
\begin{equation}\label{aaron1}
    \bbP_x \left[ |\{ 0 \leq t \leq T \colon  X_{t}  \in  \cX^{*} \}| \geq \frac{2}{5} T \right] > 0.75.
\end{equation}

Next, for set $S \subset \cX^*$, define 
\begin{equation*}
    \tilde{\tau}_{\mathrm{hit}}(S) = \min\{t \, : \, \tilde{X}_{t} \in S \}.
\end{equation*}
Denote by $\tau_{\mathrm{hit}}$ the analogous hitting time for the original chain. 
Fix an arbitrary subset $S \subset \mathcal{X}$ with $\pi(S) > 0.45$. Since $\pi(\cX^*) \geq 0.65$, we have  $\pi(S \cap \cX^*) > 0.1$. For any $x \in \cX$, we have 
\begin{align*}
    \mathbb{P}_x[\tilde{\tau}_{\mathrm{hit}}(S) \leq 4 \tmix(P|_{\cX^*})] &\geq  \mathbb{P}_x [\tilde{X}_{4 \tmix(P|_{\cX^*})} \in S] \\ 
    &\geq \pi|_{\cX^*}(S \cap \cX^*) - \|   P|_{\cX^*}^{4 \tmix(P|_{\cX^*})} (x, \cdot ) - \pi|_{\cX^*}(\cdot) \|_{\mathrm{TV}} \\
    &\geq 0.1 - 2^{-4} \\
    &\geq 0.03.
\end{align*}

Iterating, for any $S \subset \mathcal{X}$ with $\pi(S) > 0.45$, we have for $k \in \mathbb{N}$, 
\begin{equation} \label{IneqHitOnTrace}
  \max_{x \in \cX} \mathbb{P}_x [\tilde{\tau}_{\mathrm{hit}}(S) > 4 k \, \tmix(P|_{\cX^*})] \leq 0.97^{k}.
\end{equation}
Consider an integer 
$$T \geq \max \left\{ \frac{T_1}{5}, \,  \frac{2 \log(4)}{-\log(0.97)}  \tmix(P|_{\cX^*}) \right\}.$$
Using a union bound and then combining inequalities \eqref{aaron1} and \eqref{IneqHitOnTrace}, we have:
\begin{align*}
    \max_{x \in \cX} \mathbb{P}_x[\tau_{\mathrm{hit}}(S) > 5T] &\leq  \max_{x \in \cX} \mathbb{P}_x[|\{1 \leq t  \leq 5 T \, : \, X_i \notin \cX^* \}| > 3 T] +  \max_{x \in \cX} \mathbb{P}_x[\tilde{\tau}_{\mathrm{hit}}(S) > 2 T ]  \\
    &\leq 0.25 + 0.25 \\
    &\leq 0.5.
\end{align*}
Thus, for each $k \in \bbN$,
\begin{align*}
      \max_{x \in \cX} \mathbb{P}_x[\tau_{\mathrm{hit}}(S) > 5 k T] \leq 2^{-k}. 
\end{align*}
It follows that 
\begin{align*}
    \max_{x \in \mathcal{X}} \max_{S \subset \mathcal{X} \, : \, \pi(S) > 0.6} \mathbb{E}_x[\tau_{\mathrm{hit}}(S) ] \leq 10 T.
\end{align*}
Applying the main result of \cite{peres2015mixing}  completes the proof. 
\end{proof}

\section{Proofs} \label{supp:all.proofs}

\subsection{Proof of Lemma~\ref{lm:pih}}\label{supp:proof.pih}
\begin{proof}
It suffices to check the detailed balance condition. If $h(u) = u^a$, then for any $y \in \cN(x)$, 
\begin{align*}
    \pi(x)^{2a} Z_h(x) K_h(x, y)  = \pi(x)^{2a}  h\left( \frac{\pi(y)}{\pi(x)} \right)  
    = \pi(x)^a \pi(y)^a = \pi(y)^{2a} Z_h(y) K_h(y, x). 
\end{align*}
Hence, $\pi_h \propto Z_h \cdot  \pi$.  Similarly, if $h$ is a balancing function, one can use Definition~\ref{def:balance} to show that $\pi(x) Z_h(x)  K_h(x, y) =  \pi(y)  Z_h(y) K_h(y, x)$. 
\end{proof}

\subsection{Proof of Lemma~\ref{lm:tgs}}\label{supp:proof.is}
\begin{proof} 
Without loss of generality, we can assume that $x^{(0)}$ is generated from $\pi_h$, since  the limiting distribution of the estimator~\eqref{eq:sni} does not depend on the initial distribution~\citep[Chapter 22.5]{douc2018markov}.  
Let $W^{(0)}, W^{(1)}, \dots, W^{(t)}$ denote random variables such that given $x^{(k)}$, $W^{(k)}$ is an exponential random variable with mean $\omega(x^{(k)})$ and independent of everything else. Define 
\begin{align*}
\hat{f}_Q ( T_m) =    \frac{ \sum_{k=0}^m f(x^{(k)}) W^{(k)} } { T_m } , \quad \quad \text{ where } T_m = \sum_{k=0}^{m}  W^{(k)}. 
\end{align*}
Now one can see that  Proposition 5 of~\citet{deligiannidis2018ergodic} differs from our setting only in that the former assumes each $W^{(k)}$ is geometrically distributed~\citep[c.f.][ Proposition 2]{doucet2015efficient}.  
So, we can apply their argument. 
Since  $\omega = \pi / \pi_h$, by the law of large numbers and central limit theorem for ergodic Markov chains~\citep[Corollary 6]{haggstrom2007variance}, 
\begin{align*}
\frac{1}{t} \sum_{k=1}^t   \omega (x^{(k)}) \overset{a.s.}{\rightarrow} \bbE_{\pi_h} [ \omega ] =  1,   \quad \quad 
\frac{1}{\sqrt{t} } \sum_{k=1}^t f(x^{(k)}) \omega (x^{(k)})  \overset{D}{\rightarrow} N(0,   \sigma^2 ), 
\end{align*}
where 
\begin{align*}
\sigma^2 = \lim_{t \rightarrow \infty}  t^{-1} \, \Var \left(   \sum_{k=1}^t  f(x^{(k)}) \omega (x^{(k)})   \right). 
\end{align*}
It then follows from  Slutsky's theorem that $\sqrt{t} \hat{f} (t, \omega)  \overset{D}{\rightarrow} N(0,   \sigma^2)$.
Similarly, by a standard conditioning argument and treating $(f(x^{(k)}), W^{(k)})$ as a bivariate Markov chain, we obtain that 
$ \sqrt{T_m} \hat{f}_Q ( T_m)  \overset{D}{\rightarrow} N(0,   \sigma^2_{\mathrm{c}})$ as $m \rightarrow \infty$,  where 
\begin{align*}
\sigma_{\rm{c}}^2 = \lim_{t \rightarrow \infty}  t^{-1} \, \Var \left(   \sum_{k=1}^t  f(x^{(k)}) W^{(k)} \right). 
\end{align*}
But $\hat{f}_Q( T_m)$ is just the time average of the continuous-time Markov chain  $Q_h$ at time $T_m$. Thus, by standard results (see, for example,~\citet[Proposition 4.29]{aldous2002reversible}), 
\begin{equation}
\sigma_{\rm{c}}^2 \leq \frac{ \bbE_\pi[f^2]}{\gap(Q_h)}. 
\end{equation}
So it only remains to compare $\sigma^2$ with $\sigma_{\rm{c}}^2$. 
A direct calculation using conditioning yields that 
\begin{align*}
&  \bbE_{\pi_h } \left[  \sum_{k=1}^t  f(x^{(k)}) W^{(k)}  \right ]^2  -  \bbE_{\pi_h } \left[  \sum_{k=1}^t f(x^{(k)}) \omega (x^{(k)})   \right ]^2  \\
= \;&  \sum_{k=1}^t   \bbE_{\pi_h }\left[  f(x^{(k)})^2 (W^{(k)})^2 - f(x^{(k)})^2 \omega (x^{(k)})^2  \right] \\
= \;&    t \, \bbE_{\pi_h }\left[  f(X )^2  W^2 - f(X)^2 \omega (X )^2   \right] \quad \quad  \text{(where } W\mid X \sim \mathrm{Exp}(1 / \omega(X)),  \; X \sim \pi_h ) \\
= \;& t \, \bbE_\pi [f^2 \omega ]. 
\end{align*}
Hence, $ \sigma^2 = \sigma_{\rm{c}}^2 - \bbE_\pi [f^2 \omega ]$, from which the result follows. 
\end{proof}

\subsection{Proof of  Lemma~\ref{lm:main} }\label{supp:proof.main}
\begin{proof} 
We first show that it suffices to prove the claim for $\gap(P)$. 
Let $b= \max_{ x \in \cX} |Q(x, x)|$, which is finite since $|\cX| < \infty$. 
Then $P = b^{-1} Q + I$ is the transition matrix of a discrete-time Markov chain such that $  \gap(Q) =  b  \gap(P)$, which is still irreducible and reversible w.r.t. $\pi$. 
Since $x \neq \cT(x)$, we have $P (x, \cT(x) )  = b^{-1}  \,  Q(x, \cT(x))$. 
Thus, if $\gap(P) \geq \kappa(p, \alpha, \nu) \min_{x \neq x^*} P(x, \cT(x))$, the same bound holds for $\gap(Q)$ with $P$ replaced by $Q$.  
  
Our proof for $\gap(P)$ is conceptually similar to the analysis of the birth-death chain given in~\citet[Section 3]{kahale1997semidefinite}. 
Without loss of generality, we can assume that $P(x, \cT(x)) > 0$, since otherwise the spectral gap bound holds trivially. 
Set $\cT(x^*) = x^*$ to avoid ambiguity. 
Let $(\cX, \cT)$ be the bidirected graph with node set $\cX$ and edge set  
$\Edge(\cT) = \{  (x, y) \in \cX^2 \colon   x \neq y, \text{ and }  y = \cT(x) \text{ or } x = \cT(y) \}$, which implies $(x, y) \in \Edge (\cT)$ if and only if $(y, x) \in \Edge(\cT)$.  
Observe that $(\cX, \cT)$ is a tree, and thus for any $x \neq y$, there exists one unique directed path  without repeated edges that starts at $x$ and ends at $y$; denote this path by $\gamma(x, y)$. 
We use the notation $e \in \gamma$  to mean that the path $\gamma$ traverses the  edge $e$. 

Given an edge $e = (z, w) \in \Edge(\cT)$, we define its load by 
\begin{align*}
\rho(e ) = \pi(z) P(z, w) = \pi(w) P(w, z). 
\end{align*} 
The second equality holds since $P$ is reversible. Define the ``length'' of this edge by 
\begin{equation}\label{eq:length}
\ell( e ) = \left\{   \pi(z) \wedge \pi(w)  \right\}^{-q},   \quad \quad  q = \frac{\nu - \alpha}{2 \nu}. 
\end{equation}
For any directed path $\gamma$, let  $ |\gamma|_\ell = \sum_{e \in \gamma} \ell(  e)$ denote the ``length'' of the path.  
Note that for any $x \neq x^*$,  there exists some $d = d(x) < \infty$ such that $\gamma(x, x^*) = (x, \cT(x), \dots, \cT^d(x))$. It follows that 
\begin{align*}
 |\gamma(x, x^*)|_\ell =  \sum_{e \in \gamma(x, x^*)} \ell(  e)   =  \sum_{k=0}^{d-1}  \pi( \cT^k(x) )^{-q} 
 \leq  \sum_{k=0}^{d-1} \pi(x)^{- q} p^{- k \nu q} \leq \frac{ \pi(x)^{- q} }{ 1 - p^{ - \nu q} }, 
\end{align*}
where $\cT^0(x)$ denotes $x$ itself.    For any $x \neq y$, we can bound the length of $\gamma(x, y)$ by 
\begin{equation}\label{eq:lm1}
 |\gamma(x, y)|_\ell  \leq  |\gamma(x, x^*)|_\ell +  |\gamma(y, x^*)|_\ell \leq \frac{ \pi(x)^{- q}  + \pi(y)^{- q}}{ 1 - p^{ - \nu q} }. 
\end{equation}
By~\citet[Theorem 3.2.3]{saloff1997lectures}, 
\begin{equation}\label{eq:lm1b}
\gap(P)^{-1} \leq  \max_{e \in \Edge(\cT) }  \left\{  \frac{1}{ \rho(e) \ell(e)}   \sum_{(x, y) \colon e \in \gamma(x, y)}  \pi(x) \pi(y) |\gamma(x, y)|_\ell  \right\}. 
\end{equation} 

The rest of the argument is similar to the proof of Theorem 1 of~\citet{zhou2021complexity}. 
To bound the right-hand side of the above inequality, by symmetry, it suffices to consider edges $e = (z, w)$ such that $w = \cT(z)$.  
Fix an arbitrary $z \neq x^*$ and let $w = \cT(z)$. 
Let $\Anc(z) = \{ x \in \cX \colon   \cT^k(x) = z \text{ for some } k \geq 0 \} $ denote all the ``ancestors'' of $z$ (including $z$ itself). 
Recalling that $(\cX, \cT)$ is a tree, one can show that $e = (z, w) \in \gamma(x, y)$  only if $x \in \Anc(z)$ and $y \notin \Anc(z)$. 
Hence, by~\eqref{eq:lm1}, 
\begin{align*}
\sum_{(x, y) \colon e \in \gamma(x, y)}  \pi(x) \pi(y) |\gamma(x, y)|_\ell 
\leq \;& \frac{1}{ 1 - p^{ - \nu q} } \sum_{x \in \Anc(z)} \sum_{y \notin \Anc(z)}   \left\{ \pi(x)^{ 1- q} \pi(y) +  \pi(x)\pi(y)^{1 - q} \right\} \\
\leq \;& \frac{2}{ 1 - p^{ - \nu q} }  \sum_{x \in \Anc(z)} \sum_{y \in \Anc(x^*)}   \pi(x)^{ 1- q}  \pi(y)^{1 - q}. 
\end{align*}
The assumption $|\cN(x)| \leq p^{\alpha}$  and the reversibility of $P$ imply that $|\{y \in \cX \colon \cT(y) = x \} | \leq p^\alpha$.
Let $\cT^{-k}(z) = \{x \colon \cT^k(x) = z, \cT^{k-1}(x) \neq z \}$. Then, $\pi(\cT^{-k}(z))^c \leq p^{(\alpha - \nu c) k} \pi(z)^c$ for any $c > 0$. Using $q = (\nu - \alpha) / 2\nu$, we find that 
\begin{equation}\label{eq:a1}
 \sum_{x \in \Anc(z) } \pi(x)^{1 - q} \leq \sum_{k=0}^\infty \pi(z)^{1 - q} p^{[ \alpha - \nu (1 - q)] k} \leq \frac{\pi(z)^{1 - q} }{1 - p^{\alpha - \nu (1 - q)}} = \frac{ \pi(z)^{1 - q} }{1 - p^{- (\nu - \alpha)/2}}. 
\end{equation} 
It follows that 
\begin{align*}
    \sum_{(x, y) \colon e \in \gamma(x, y)}  \pi(x) \pi(y) |\gamma(x, y)|_\ell \leq \frac{2\pi(z)^{1 - q}} { \left\{1 - p^{- (\nu - \alpha)/2} \right\}^3}. 
\end{align*}
Plugging this inequality into~\eqref{eq:lm1b}, we obtain that 
\begin{align*}
\gap(P)^{-1} \leq   \max_{z \neq x^*} \frac{2 }{ \{1 - p^{-(\nu - \alpha)/2}  \}^3 } P(z, \cT(z))^{-1}, 
\end{align*}
which yields the asserted spectral gap bound.  
\end{proof}

\subsection{Proof of Theorem~\ref{th:one} }\label{supp:th1}
\begin{proof}
From Lemma~\ref{lm:pih} we know that for a balancing function $h$, we have $\pi_h \propto \pi Z_h$, which can be equivalently expressed as $\pi_h (x) = \pi(x) Z_h(x) / \bbE_\pi [Z_h ]$.  
Hence,  for any $y \neq x$, 
\begin{equation}\label{eq:Q1}
Q_h(x, y) = \frac{\pi_h(x)}{\pi(x)} \frac{ h( \pi(y) / \pi(x) ) }{Z_h(x)}
=  \frac{ h( \pi(y) / \pi(x) ) }{ \bbE_\pi [Z_h ] }. 
\end{equation}
If $h$ is non-decreasing, we have $Q_h(x, \cT(x))  \geq h(p^\nu) / \bbE_\pi[Z_h]$ under Assumption~\ref{A1}. Let $I_h(x, y) = \pi(x)h ( \pi(y) / \pi(x) ) = I_h(y, x)$. 
Then, 
\begin{equation}\label{eq:EZ}
\bbE_\pi[Z_h] = \sum_{x \in \cX} \pi(x) Z_h(x) = \sum_{x \in \cX} \sum_{y \in \cN(x)}  I_h(x, y). 
\end{equation}
Since $\cN$ is symmetric,  $\bbE_\pi[Z_h]$ is twice the sum of $I_h(x, y)$ over all unordered pairs of neighbors. 
Now we bound $\bbE_\pi[Z_h]$ for the three choices of $h$ separately, from which the asserted bounds on $\gap(Q_h)$ follow. 

\medskip 
\noindent \textit{Case 1: $h(u) = 1 + u$.}   We have $I_h(x, y) = \pi(x) + \pi(y)$. 
Since each $x$ has at most $p^\alpha$ neighbors,  
\begin{align*}
\bbE_\pi[Z_h] = 2 \sum_{x \in \cX}   |\cN(x)| \pi(x) \leq 2 p^\alpha. 
\end{align*}
  
\medskip
\noindent \textit{Case 2: $h(u) = 1 \wedge u$.} We have $I_h(x, y) = \pi(x) \wedge \pi(y)$. 
For any $x \neq x^*$, $\pi(x)$ can appear in the  summation in~\eqref{eq:EZ} at most $2 |\cN(x)|$ times. 
But $\pi (x^*) $ cannot appear in the summation since $x^*$ is the mode. 
By a calculation similar to~\eqref{eq:a1},  we find that $\pi(x^*) \geq 1 - p^{\alpha - \nu}$  under Assumption~\ref{A1}.  Thus, 
\begin{align*}
\bbE_\pi[Z_h]  \leq 2  \sum_{x \neq x^*} |\cN(x) | \pi(x) \leq 2 p^\alpha (1 - \pi(x^*)) \leq 2 p^{2 \alpha - \nu}. 
\end{align*} 

\medskip
\noindent \textit{Case 3: $h(u) = \sqrt{u}$.}  We have $I_h(x, y) = \sqrt{ \pi(x)\pi(y) } \leq ( \pi(x) + \pi(y) )/ 2$. 
Applying this inequality to any pair of neighbors that does not involve $x^*$, we obtain from~\eqref{eq:EZ} that
\begin{align*}
\bbE_\pi[Z_h]  \leq  2 \sum_{x \neq x^*} |\cN(x)| \pi(x) + 2 \sum_{y \in \cN(x^*)}  \sqrt{ \pi(x)\pi(y) }. 
\end{align*}
Since for any $x \neq x^*$, $\pi(x) \leq p^{-\nu}$ under Assumption~\ref{A1}, we have 
\begin{align*}
\bbE_\pi[Z_h]  \leq  2p^{2 \alpha - \nu} + 2 p^{\alpha - \nu/2},
\end{align*} 
which completes the proof. 
\end{proof}

\subsection{Proof of Theorem~\ref{th:two}}\label{supp:th2}
\begin{proof}
Let $b= \max_{ x \in \cX} |Q_h(x, x)|$. Then $P = b^{-1} Q_h + I$ is the transition matrix of a discrete-time Markov chain such that $  \gap(Q_h) =  b  \gap(P)$.   
Treating $\cX^*$ as a single state denoted by $x^*$, define a transition matrix $\bP$ with  state space $(\cX \setminus \cX^*) \cup \{x^*\}$ by 
\begin{equation}\label{eq:def.Pbar}
\begin{aligned}
    \bP(x, x') = \left\{\begin{array}{cc}
        P (x, x'), &  \text{ if } x\neq x^*, x' \neq x^* \\
    \sum_{y \in \cX^*} P(x, y),    &  \text{ if } x' = x^*, x \neq x^* \\
    \pi(\cX^*)^{-1} \sum_{y \in \cX^*} \pi(y) P(y, x')    &  \text{ if } x = x^*, x'\neq x^* \\
    \pi(\cX^*)^{-1} \sum_{y, w \in \cX^*} \pi(y) P(y, w)    &  \text{ if } x = x' = x^*. \\
    \end{array}
    \right.
\end{aligned}
\end{equation} 
Then $\bP$ is reversible with respect to $\bpi$ defined by $\bpi(x) = \pi(x)$ for $x \neq x^*$ and $\bpi(x^*) = \pi(\cX^*)$. Let $\bar{\cN}$ denote the neighborhood mapping on $\bar{\cX}$ induced by $\cN$. Then, $|\cN(x)| \leq p^\alpha$ for any $x \neq x^*$, and $|\cN(x^*)| \leq M p^\alpha$.
The mapping $\bar{\cT} \colon \bar{\cX} \rightarrow \bar{\cX}$ induced by $\cT$ can be defined by
\begin{align*}
    \bar{\cT}(x) =\left\{\begin{array}{cc}
        x^* &  \text{ if } x = x^*, \\
       \cT(x)  & \text{ if } x \neq x^*, \cT(x) \notin \cX^*, \\
        x^*  & \text{ if } x \neq x^*, \cT(x) \in \cX^*, \\
    \end{array}
    \right.
\end{align*}
Observe that $\bpi(\bar{\cT}(x)) \geq p^\nu \bpi(x)$ for any $x \neq x^*$. Hence, we can bound $\gap(\bP)$ by the same argument used to prove Lemma~\ref{lm:main}. The only step that needs to be modified is~\eqref{eq:a1}. Since 
\begin{equation}\label{eq:a2}
\bpi( \bar{\cT}^{-1}(x^*) )^c \leq M p^\alpha \left( \frac{\bpi(x^*)}{M p^\nu} \right)^c =  M^{1-c} p^{\alpha - \nu c} \bpi(x^*)^c,
\end{equation}
for any $c > 0$, one can show that it suffices to multiply the bound in~\eqref{eq:a1} by $M^q$ where $q = (\nu - \alpha) / 2\nu$. It follows that 
\begin{equation}\label{eq:bd.barP}
    \gap(\bP) \geq \frac{\kappa}{M^{(\nu - \alpha)/\nu} } \min_{x \neq x^*} \bP(x, \bar{\cT}(x)) 
    = \frac{\kappa}{M^{(\nu - \alpha)/\nu} } \min_{x \notin \cX^*} P(x, \cT(x)) 
    \geq \frac{\kappa}{M} \min_{x \notin \cX^*} P(x, \cT(x))  , 
\end{equation}
where $\kappa = \kappa(p, \alpha, \nu)$ is as given in Lemma~\ref{lm:main}. 

Next, define $P^* \colon \cX^* \times \cX^* \rightarrow [0, 1]$ as the restriction of $P$ to $\cX^*$; that is, $P^*(x, x') = P(x, x')$ if $x \neq x'$, and $P^*(x, x) = 1 - \sum_{x' \in \cX^* \setminus \{x\}} P(x, x')$.
By Theorem 1 of~\citet{jerrum2004elementary}, we have 
\begin{equation}\label{eq:jerrum}
\gap(P) \geq \min \left\{ \frac{ \gap(\bP) }{3},  \;  \frac{\gap(\bP) \gap(P^*) }{3 \sgamma + \gap(\bP)} \right\}, 
\end{equation}    
where $\sgamma = \max_{x \in \cX^*}  P(x, \cX \setminus \cX^*)$. 

As in the proof of Theorem~\ref{th:one},  we have 
\begin{equation}\label{eq:qh}
    Q_h(x, x') = \frac{\pi_h(x')}{\pi(x)} \frac{ h( \pi(x')/\pi(x)) }{ Z_h(x)} = \frac{ h( \pi(x')/\pi(x)) }{ \bbE_\pi [Z_h] }. 
\end{equation}
Using $P = b^{-1} Q_h + I$, we obtain the following bounds on $\gap(\bP), \gap(P^*)$ and $\sgamma$ 
\begin{align}
    \gap(\bP) \geq \;& \frac{\kappa  h(p^{\nu})}{b M \bbE_\pi[Z_h] }, \label{eq:bd1} \\
    \gap(P^*) \geq \;& \frac{h(1)}{b M(M-1) \bbE_\pi[Z_h]}, \label{eq:bd2} \\
    \sgamma \leq \;& \frac{ p^\alpha h(p^{-\nu})}{b \bbE_\pi[Z_h] } \label{eq:bd3}. 
\end{align}
The first inequality follows from~\eqref{eq:bd.barP}, the monotonicity of $h$ and the fact that $\bpi(\bar{\cT}(x)) \geq p^\nu \bpi(x)$ for any $x\neq x^*$. 
To prove the second, recall that Condition (iv) in Assumption~\ref{A2} implies that $P^*$ is irreducible. Hence, between any $x, x' \in \cX^*$, there is a path on $\cX^*$ with length at most $M-1$.  Further, $P^*$ is reversible with respect to the uniform measure on $\cX^*$, and if $x, x' \in \cX^*$ and $x' \in \cN(x)$, we have $P^*(x, x') = P(x, x') \geq h(1)/ ( b \bbE_\pi[Z_h])$. Then,~\eqref{eq:bd2} follows from the standard canonical path method~\citep{sinclair1992improved}. Lastly,~\eqref{eq:bd3} can be proved by noting that $\pi(x) \geq p^\nu \pi(x')$ for any $x \in \cX^*, x' \notin \cX^*$. Plugging~\eqref{eq:bd1},~\eqref{eq:bd2} and~\eqref{eq:bd3} into~\eqref{eq:jerrum}, we get 
\begin{align*}
    \gap(Q_h) = b \gap(P) 
    \geq \frac{ \kappa h(p^\nu) h(1) }{3 M (M-1) \bbE_\pi[Z_h] \left\{M p^\alpha h(p^{-\nu}) + h(p^\nu) \right\} }.
\end{align*} 
Since $h$ is a balancing function, $h(p^\nu) = p^\nu h(p^{-\nu}) > p^\alpha h(p^{-\nu})$. Hence, $M p^\alpha h(p^{-\nu}) + h(p^\nu) < (Mp^{\alpha - \nu} + 1) h(p^\nu)$ and the above bound simplifies to
\begin{align*}
       \gap(Q_h) \geq \frac{ \kappa  h(1) }{3 (Mp^{\alpha - \nu} +1) M (M-1) \bbE_\pi[Z_h]  }.
\end{align*}
To complete the proof, we  bound $\bbE_\pi [Z_h]$  for each choice of $h$ using~\eqref{eq:EZ}. 
Let $\pi_0 = \pi(x)$ for any $x \in \cX^*$. 
Note that by letting $c=1$ in~\eqref{eq:a2}, we can show that $\bpi(x^*) = M \pi_0 \geq 1 - p^{\alpha - \nu}$. 

\medskip 
\noindent \textit{Case 1: $h(u) = 1 + u$.} The same bound $\bbE_\pi [Z_h] \leq 2p^\alpha$ holds.     

\medskip
\noindent \textit{Case 2: $h(u) = 1 \wedge u$.}  
Observe that $I_h(x, y)  = \pi_0$ only when both $x, y \in \cX^*$,  which happens at most $M(M-1)$ times in the  summation in~\eqref{eq:EZ}. 
Hence, 
\begin{align*}
\bbE_\pi[Z_h]  \leq 2  \sum_{x \notin  \cX^*} |\cN(x) | \pi(x) + M(M-1) \pi_0 \leq 2 p^{2 \alpha  - \nu}  + M-1. 
\end{align*}

\medskip
\noindent \textit{Case 3: $h(u) = \sqrt{u}$.} To bound $I_h(x, y) = \sqrt{ \pi(x) \pi(y)}$, we consider three subcases according to whether $x$ and $y$ are in $\cX^*$. 
First, using $\sqrt{\pi(x) \pi(y)} \leq (\pi(x) + \pi(y))/2$, we find  that 
\begin{align*}
\sum_{(x, y) \in (\cX \setminus \cX^*)^2 \colon y \in \cN(x)}  \sqrt{\pi(x) \pi(y)}
\leq \;& \sum_{x \in \cX \setminus \cX^*  } |\cN(x)|   \pi(x)  \leq  p^{2 \alpha - \nu}.  
\end{align*}
If $x \in \cX^*$ and $y \notin \cX^*$, we have $\pi(y) \leq p^{-\nu} \pi(x)$. Hence,
\begin{align*}
\sum_{ x \in \cX^*, \, y \notin \cX^*}  \sqrt{\pi(x) \pi(y)}
\leq \;&  \sum_{ x \in \cX^*, \,  y \in \cN(x)} \pi(x) p^{-\nu / 2} \leq p^{\alpha - \nu / 2}. 
 \end{align*}
Using $\pi_0 \leq 1 / M$, we finally obtain that 
\begin{align*}
\bbE_\pi[Z_h] \leq  p^{2 \alpha - \nu} + 2p^{\alpha - \nu / 2} + M - 1, 
\end{align*} 
which completes the proof. 
\end{proof}

\subsection{Proof of Theorem~\ref{th:mix}}\label{supp:mix.bound}
\begin{proof}[Proof of Theorem~\ref{th:mix}]
Define a transition matrix $K_0$ with state space $\cX$ by 
\begin{align*}
    K_0 (x, x') = \left\{\begin{array}{cc}
        \frac{Q_h(x, x')}{ \sum_{y \neq x}  Q_h(x, y)},  &   \text{ if } x \neq x', \\
        0  & \text{ if } x = x'. 
    \end{array}
    \right.
\end{align*}
Let $K_0 |_{\cX^*}$ be the trace of $K_0$ on $\cX^*$.  
Fix an arbitrary $b \geq  2 \max_{ x \in \cX} |Q_h(x, x)|$, and let $P^b = b^{-1} Q_h + I$ be the transition matrix of a discrete-time Markov chain with $P^b(x, x) \geq 1/2$ for each $x$. 
If $x, x' \in \cX^*$ and $x \neq x'$, then 
\begin{align*}
    P^b|_{ \cX^*}(x, x') = (1 - P^b(x, x) ) K_0|_{\cX^*}(x, x') = - \frac{Q_h(x, x)}{b} K_0|_{\cX^*}(x, x'). 
\end{align*}
This shows that if $P^b(x, x') \geq \delta / b$ for $b = 2 \max_{ x \in \cX} |Q_h(x, x)|$, then we have $P^b(x, x') \geq \delta / b$ for any $b \geq 2 \max_{ x \in \cX} |Q_h(x, x)|$. Hence, we can assume $b$ is chosen sufficiently large, and then~\citet[Theorem 20.3]{levin2017markov} implies that there exists a universal constant $C'$ such that $  \tmix(Q_h) \leq C'   b^{-1}  \tmix(P^b)$. 
In the rest of the proof, we simply use $P$ to denote $P^b$. 

To bound $\tmix(P)$, we apply Lemma~\ref{LemmaDecompGeneric}  along with the bounds appearing in the proof of Theorem~\ref{th:two}.  We must first choose our decomposition $\cX = \sqcup_{i=1}^{M} \cX_{i}$. Let $\cX^* = \{ x_1^*, \dots, x_M^* \}$. 
For each $x \notin \cX^*$, let $d(x) = \min\{ k \colon  \cT^k(x)  \in \cX^* \}$, which is finite by Assumption~\ref{A2}. 
For each $k \in \{1, \dots, M\}$, define $\cX_k =  \{x_k^*\}  \cup  \{ x \notin \cX^* \colon  \cT^{d(x)}(x) = x_k^* \}$.  Then, by construction, $\cX_1, \dots, \cX_M$ are disjoint, and $\cup_{k=1}^M \cX_k = \cX$.   
Let $P_k$ be the restriction of $P$ to $\cX_k$, and $P|_{\cX_k}$ be the trace of $P$ on $\cX_k$. Both are reversible with respect to the measure $\pi_k$ defined by $\pi_k(x) = \pi(x) / \pi(\cX_k)$ for each $x \in \cX_k$. Hence, by Peskun's theorem, $\gap(P|_{\cX_k}) \geq \gap(P_k)$. 
Let $\cN_k$ be the restriction of $\cN$ to $\cX_k$ such that $\cN_k(x) = \cN(x) \cap \cX_k$ for each $x \in \cX_k$. 
Observe that for any $x \in  \cX_k \setminus \{x_k^*\}$, we have $\cT(x) \in \cX_k$. Thus, the quadruple $(\cX_k, \cN_k, \pi_k, \cT)$ satisfies Assumption~\ref{A1} with the same constants $\nu, \alpha$. 
Hence, by Theorem~\ref{th:one}, we have 
\begin{align*}
\gap(P|_{\cX_k}) \geq \gap(P_k) \geq \kappa \min_{x \in \cX_k \setminus \{x_k^*\}} P_k(x, \cT(x) ) =  \kappa \min_{x \in \cX_k \setminus \{x_k^*\}} P(x, \cT(x) ), 
\end{align*}
where $\kappa = \kappa(p, \alpha, \nu)$ is as given in Lemma~\ref{lm:main}.  It follows that 
\begin{equation}\label{IneqThm2AnalogueRest}
\gap_{\rm{min}} = \min_{1 \leq k \leq M}\gap(P|_{\cX_k}) \geq  \kappa  \min_{x \in  \cX \setminus \cX^*} P(x, \cT(x) )  \geq    \frac{h(p^\nu)}{b \bbE_\pi[Z_h]}. 
\end{equation}

By a calculation analogous to~\eqref{eq:a2}, one can show that $\pi(x_i^*) / \pi(\cX_i) \rightarrow 1$ as $p^{\nu - \alpha} \rightarrow \infty$, and thus the condition $\min_{1 \leq i \leq M}\pi(x_i^*) / \pi(\cX_i) \geq 0.65$ in Lemma~\ref{LemmaDecompGeneric} is satisfied if $p^{\nu - \alpha}$ is sufficiently large.  
The bound on $\gap(P^*)$ appearing in the proof of Theorem~\ref{th:two} applies to $P|_{ \cX^*}$ with only two modifications. First, the lower bound on $P|_{\cX^*}(x,x') \geq \delta / b$ is obtained directly from the assumption. Second, the stationary probability of $P|_{\cX^*}$ is bounded by $\pi|_{\cX^*} (x) \geq 1 / (B M)$ (in Theorem~\ref{th:two} we assume $B=1$).  This gives:
\begin{align*}
    \gap(P|_{\cX^*}) \geq \frac{\delta / b}{ B M(M-1)}.
\end{align*}
Thus, by~\citet[Theorem 20.6]{levin2017markov}, 
\begin{equation} \label{IneqThm2AnalogueTrace}
    \tmix(P|_{ \cX^*}) \leq (b/\delta) B M(M-1) \log(4 B M).
\end{equation} 
Combining inequalities \eqref{IneqThm2AnalogueRest}, \eqref{IneqThm2AnalogueTrace}, and applying Lemma~\ref{LemmaDecompGeneric}, we obtain the bound for $\tmix(Q_h)$. By~\citet[Lemma 4.23]{aldous2002reversible}, $\gap(Q_h)^{-1} \leq  \tmix(Q_h)$, which concludes the proof. 
\end{proof}

\subsection{Proof of Theorem~\ref{th:decomp}}\label{supp:decomp}
\begin{proof}
The proof is similar to that of Theorem~\ref{th:two}. 
Let $b= \max_{ x \in \cX} |Q_h(x, x)|$ and $P = b^{-1} Q_h + I$. 
Let $P_\cG$ denote the Markov chain induced by $\cG$ on $\cY$, which is defined by
\begin{align*}
    P_{\cG}(y, y') =\pi_{\cG}(y)^{-1}  \sum_{x \in \cX_y} \sum_{x' \in \cX_{y'}} \pi(x) P(x, x'), \quad \forall \, y, y' \in \cY, 
\end{align*}
It is easy to verify that $P_{\cG}$ is irreducible and reversible with respect to $\pi_{\cG}$. 
Since $\pi_\cG$ on $\cY$ satisfies Assumption~\ref{A1},
we can apply Lemma~\ref{lm:main} to obtain that\footnote{Though in Assumption~\ref{A1} we require $|\cN(x)| \leq p^\alpha$,  the proof of Lemma~\ref{lm:main} only uses $|\{x' \colon \cT(x') = x \}| \leq p^\alpha$.}  
\begin{align*}
    \gap( P_\cG) \geq \kappa(p, \alpha, \nu) \min_{y \neq y^*}  P_\cG(y, \cT(y)) \eqqcolon \bar{\lambda}. 
\end{align*}
For each $y \in \cY$, let $ P_y$ be the restricted Markov chain on $\cX_y$ defined by 
\begin{align*}
     P_y(x, x') = \left\{\begin{array}{cc}
        P(x, x') &  \text{ if } x, x' \in \cX_y \text{ and } x \neq x' \\
        1 - \sum_{w \in \cX_y \setminus \{x\}} P(x, w) & \text{ if } x = x'  \in \cX_y. 
    \end{array}    \right.
\end{align*}
Let $\lambda_{\rm{min}} = \min_{y \in \cY} \gap( P_y)$. 
By Theorem 1 of~\citet{jerrum2004elementary}, we have 
\begin{equation}\label{eq:jerrum2}
\gap(P) \geq \min \left\{ \frac{ \bar{\lambda} }{3},  \;  \frac{\bar{\lambda} \lambda_{\rm{min}} }{3 \bar{\gamma} + \bar{\lambda}} \right\}, 
\end{equation}    
where $\bar{\gamma} = \max_{x \in \cX} \sum_{x' \colon \cG(x') \neq \cG(x)} P(x, x')$. By~\eqref{eq:qh}, we can bound $\bar{\gamma}$ by
 \begin{equation}\label{eq:bd.beta}
     \bar{\gamma}  \leq  \max_{x \in \cX} (1 - P(x,x)) = \max_{x \in \cX}  \frac{ \pi_h(x) }{b \, \pi(x)}   =  \frac{\max_{x \in \cX}  Z_h(x) }{b \, \bbE_\pi[Z_h]}. 
\end{equation}

Observe that for any $y \neq y^*$ and $x \in   \cX_y( \cT(y), \tilde{\nu})$, there exists  $x' \in \cX_{\cT(y)} \cap \cN(x)$ such that  $\pi(x') \geq p^{\tilde{\nu}} \pi(x)$. 
It then follows from condition (iii) that for any $y \neq y^*$,
\begin{align*}
      P_{\cG}(y, \cT(y)) \geq \;& \pi_{\cG}(y)^{-1}  \sum_{x \in  \cX_y( \cT(y), \tilde{\nu})} \sum_{x' \in \cX_{\cT(y)} } \pi(x) P(x, x')  \\
      \geq \;& \pi_{\cG}(y)^{-1}  \sum_{x \in  \cX_y( \cT(y), \tilde{\nu})}   \pi(x)  \frac{ h(p^{\tilde{\nu}}) }{b \bbE_\pi [Z_h] } \\
      \geq \;&   \frac{\varepsilon  h(p^{\tilde{\nu}}) }{b \bbE_\pi [Z_h] }. 
\end{align*}
Hence, we obtain the following bound on $\bar{\lambda}$.
\begin{equation}\label{eq:bar.lambda}
\bar{\lambda} \geq \frac{ \kappa  \varepsilon  h(p^{\tilde{\nu}}) }{b \bbE_\pi [Z_h] }, 
\end{equation}
where $\kappa = \kappa(p, \alpha, \nu)$. 
Plugging~\eqref{eq:bd.beta} and~\eqref{eq:bar.lambda} into~\eqref{eq:jerrum2} and using $\gap( P_y) = b^{-1} \gap(Q_y)$, we obtain that 
\begin{align*}
  \gap(Q_h) = b \gap(P) \geq \min  \left\{ \frac{ \kappa  \varepsilon  h(p^{\tilde{\nu}}) }{3   \bbE_\pi [Z_h] }, \;  \frac{ \kappa  \varepsilon  h(p^{\tilde{\nu}}) \min_{y \in \cY} \gap(Q_y)}{\kappa  \varepsilon  h(p^{\tilde{\nu}}) + 3 \max_{x \in \cX} Z_h(x) } \right\}, 
\end{align*} 
which proves the result. 
\end{proof}
 
\end{document}